%% file: main.tex
\documentclass[10pt]{iopart}
\pdfoutput=1		
\usepackage{iopams}

\usepackage{graphicx}
\usepackage{xcolor} 
\usepackage{accents} 
\usepackage{subfig}
\usepackage{amsthm} 

\theoremstyle{plain}
\newtheorem{definition}{Definition}[section]

\usepackage{multirow}
\usepackage{array}

\usepackage{algorithm,algpseudocode}

\usepackage{cite}

\floatname{algorithm}{Algorithm}

\newcommand{\eqref}[1]{(\ref{#1})}
\usepackage{amsopn}
\DeclareMathOperator{\sinc}{sinc}
\DeclareMathOperator{\DFT}{DFT}
\DeclareMathOperator{\vertcat}{vertcat}

\DeclareMathOperator{\flipud}{flipud}

\usepackage{etoolbox}
% ...
\makeatletter
\def\@mkboth#1#2{}
\newlength\appendixwidth
\preto\appendix{\addtocontents{toc}{\protect\patchl@section}}
\newcommand{\patchl@section}{%
	\settowidth{\appendixwidth}{\textbf{Appendix }}%
	\addtolength{\appendixwidth}{1.5em}%
	\patchcmd{\l@section}{1.5em}{\appendixwidth}{}{\ddt}%
}
\makeatother

\mathchardef\ordinarycolon\mathcode`\:
\mathcode`\:=\string"8000
\begingroup \catcode`\:=\active
\gdef:{\mathrel{\mathop\ordinarycolon}}
\endgroup
\newcommand{\coloneqq}{:=}

\newlength{\dhatheight}
\newcommand{\doublehat}[1]{%
	\settoheight{\dhatheight}{\ensuremath{\hat{#1}}}%
	\addtolength{\dhatheight}{-0.2ex}%
	\hat{\vphantom{\rule{1pt}{\dhatheight}}%
		\smash{\hat{#1}}}}

\newlength{\dtildeheight}

\newcommand{\ssymbol}[1]{^{\@fnsymbol{#1}}}
	
\newtheorem{theorem}{Theorem}

\newtheorem{corollary}{Corollary}[theorem]

\begin{document}
	
\title[]{The Broken Ray Transform: Additional Properties and New Inversion Formula}
\author{Michael R. Walker II, Joseph A. O'Sullivan}
\address{Preston M. Green Department of Electrical and Systems Engineering, Washington 
	University in Saint Louis, Saint Louis, Missouri, USA}
\ead{mwalkerii@wustl.edu, jao@wustl.edu}

\begin{abstract}
The significance of the broken ray transform (BRT) is due to its occurrence in a number of modalities spanning optical, x-ray, and nuclear imaging. When data are indexed by the scatter location, the BRT is both linear and shift invariant. Analyzing the BRT as a linear system provides a new perspective on the inverse problem. In this framework we contrast prior inversion formulas and identify numerical issues. This has practical benefits as well. We clarify the extent of data required for global reconstruction by decomposing the BRT as a linear combination of cone beam transforms. Additionally we leverage the two dimensional Fourier transform to derive new inversion formulas that are computationally efficient for arbitrary scatter angles. Results of numerical simulations are presented.
\end{abstract}


\section{Introduction}
\input{introduction}

The remainder of this document is organized as follows. In Section \ref{sec:analysis} we provide new analysis of the BRT from a linear systems perspective. In Section \ref{sec:algorithms} we present new numeric algorithms demonstrating some benefits of the preceding analysis. In Section \ref{sec:results} we present results of numerical simulations. Finally, some conclusions are given in Section \ref{sec:discussion}.

\section{New Analysis of the Broken Ray Transform}\label{sec:analysis}
To exploit the benefits of linear system analysis, we first derive the two-dimensional Fourier transform of the BRT data. Since the BRT is LSI, we expect the result to have a specific form. We can decompose the transform of the data into the product of two terms: the transform of the image and the transform of the system function. From the transform of the system function, several insights are directly available.

We consider an absolutely integrable image with bounded support. We define a closed, bounded, convex set $C\subset\mathbb{R}^2$, which we use to window the image
\begin{equation}
\mu_C(x) = \left\{
\begin{array}{ll}
\mu(x),  & \mbox{for } x\in C \\
0, & \mbox{otherwise.}
\end{array}
\right.
\end{equation}
We use $\doublehat{\mu}_C(w)$ to represent the two-dimensional Fourier transform of the image. Since the BRT is simply a linear combination of CBTs, we first define the two-dimensional Fourier transform of the CBT data
\numparts
\begin{eqnarray}
\doublehat{b}_\theta(w)&=\mathcal{F}^2\left\{(B \mu_C)(x,\theta)\right\}\\
&= \doublehat{\mu}_C(w)\left[\frac{-1}{i2\pi w\cdot\theta} + \frac{1}{2}\delta(w\cdot\theta)\right]\label{eq:cbt_ft}.
\end{eqnarray}
\endnumparts
The details of this derivation are in \ref{sec:cbt_ft}. The two-dimensional Fourier transform of the BRT data is therefore
\numparts
\begin{eqnarray}
\fl\doublehat{g}_{i,j}(w) &= \mathcal{F}^2\left\{(G \mu_C)(x,\theta_i,\theta_j)\right\}\\
\fl &= \mathcal{F}^2\left\{(B \mu_C)(x,\theta_i)+(B \mu_C)(x,\theta_j)\right\}\\
\fl &= \doublehat{\mu}_C(w)\left[\frac{-w\cdot\left(\theta_i+\theta_j\right)}{i2\pi \left(w\cdot\theta_i\right)\left(w\cdot\theta_j\right)}+\frac{1}{2}\delta\left(w\cdot\theta_i\right)+\frac{1}{2}\delta\left(w\cdot\theta_j\right)\right].\label{eq:brt_ft}
\end{eqnarray}
\endnumparts
Indeed, this result can be decomposed into the product of two terms. The bracketed term represents the two-dimensional Fourier transform of the BRT system function. For convenience, we will frequently reference a portion of this term
\begin{equation}
\doublehat{h}_{i,j}(w)\coloneqq \frac{-w\cdot\left(\theta_i+\theta_j\right)}{i2\pi \left(w\cdot\theta_i\right)\left(w\cdot\theta_j\right)}.\label{eq:h}
\end{equation}
We emphasize $\doublehat{h}_{i,j}(w)$ is not the transform of the BRT forward operator as the delta functions have been excluded.

The expression \eqref{eq:brt_ft} highlights some challenges with BRT inversion. We observe singularities at $w\cdot\theta_i=0$ and $w\cdot\theta_j=0$. At these frequencies, finite $\doublehat{\mu}_C(w)$ does not guarantee finite $\doublehat{g}_{i,j}(w)$. As a consequence, the data may have unbounded support. 

Additionally \eqref{eq:brt_ft} demonstrates zeros in the forward operator. We define the set $\Theta_{i,j}\subset \mathbb{R}^2$ as
\begin{equation}
\Theta_{i,j} \coloneqq \left\{w : w\cdot\left(\theta_i+\theta_j\right)=0,\,w\cdot w>0 \right\}.
\end{equation}
For all $w\in\Theta_{i,j}$, we have $\doublehat{g}_{i,j}(w)=0$ for all $\doublehat{\mu}_C(w)$. In this way the BRT has a non-trivial nullspace. The zeros are limited to a line, and so the nullspace does not include images with bounded support. This does not preclude exact analytic reconstruction of images with bounded support. However, this is problematic for numeric reconstruction. We have arrived at these observations from a linear systems perspective. Similar observations were previously made applying microlocal analysis to the BRT \cite{Sherson2015}.

A Fourier representation of the image is found multiplying both sides of \eqref{eq:brt_ft} by the inverse of \eqref{eq:h}
\begin{equation}
\doublehat{\mu}_{C}(w) = \doublehat{g}_{i,j}(w)\frac{-i2\pi \left(w\cdot\theta_i\right)\left(w\cdot\theta_j\right)}{ w\cdot\left(\theta_i+\theta_j\right)},\quad\forall\,w\notin \Theta_{i,j}.\label{eq:brt_inv_fourier}
\end{equation}
Justification for removing the delta functions is given in \ref{sec:brt_inverse_fourier}. Using \eqref{eq:brt_inv_fourier} alone, we cannot recover $\doublehat{\mu}_{C}(w)$ for $w\in \Theta_{i,j}$. According to \eqref{eq:brt_ft}, $\doublehat{g}_{i,j}(w)=0$, for all $w\in\Theta_{i,j}$,  which leaves \eqref{eq:brt_inv_fourier} indeterminate. This can be resolved imposing boundary conditions on $\mu_C(x)$ or, equivalently, continuity of $\doublehat{g}_{i,j}(w)$  (i.e. applying L'H\^{o}pital's rule).

In the parlance of linear systems analysis, this inversion formula comprises two lines of zeros and one line of pole. The zeros are associated with a directional derivatives, and the poles are associated with integration. Inverting this process leads to the reconstruction formulas
\numparts
\begin{eqnarray}
\mu_C(x)&= \frac{1}{\|\theta_i+\theta_j\|}\frac{d}{d\theta_i}\frac{d}{d\theta_j}\int_{0}^\infty g_{i,j}\left(x + s\frac{\left(\theta_i+\theta_j\right)}{\|\theta_i+\theta_j\|}\right)ds\\
&= \frac{-1}{\|\theta_i+\theta_j\|}\frac{d}{d\theta_i}\frac{d}{d\theta_j}\int_{-\infty}^0 g_{i,j}\left(x + s\frac{\left(\theta_i+\theta_j\right)}{\|\theta_i+\theta_j\|}\right)ds.
\end{eqnarray}
\endnumparts
A detailed derivation is given in \ref{sec:brt_inverse_fourier}. This is a generalization of previous inversion formulas derived by other means \cite{Florescu2011, Gouia2014, Sherson2015, Ambartsoumian2019}. Sherson was the first to recognize the symmetry in this expression\cite{Sherson2015} which is useful for numeric reconstructions from noisy data. For finite data the integrals are applied over different lengths and different noise realizations. The two results can be combined to minimize variance in the reconstruction.

The delta functions in \eqref{eq:brt_ft} present some obvious challenges. This implies images with bounded support do not guarantee data with bounded support. This begs the question: what extent of data is necessary for reconstruction? We address this in Section \ref{sec:datasupport}. Additionally, these present numerical challenges when computing the Fourier transform from sampled data. We will address this in Section \ref{sec:modulation}. Combining these results we obtain Fourier-based inversion formulas in Section \ref{sec:reconstruct}.

\subsection{Complete Representation of Data With Infinite Support}\label{sec:datasupport}
\input{analysis_support}

\subsection{Filtering Unbounded Support of the Data}\label{sec:modulation}
\input{analysis_modulation}

\subsection{Image Reconstruction from BRT Data with Bounded Support}\label{sec:reconstruct}
\input{analysis_reconstruct}

\section{Numeric Algorithms}\label{sec:algorithms}
In application we must reconstruct images from sampled data. Our analysis of the BRT from a linear systems perspective extends easily to sampled data. We demonstrate this with two new algorithms. First we provide an algorithm for extending CBT data. This is motivated by the work in Section \ref{sec:datasupport}. For a broad class of problems this can be applied to BRT data and therefore facilitates use of the filtering methods of Section \ref{sec:modulation}. For brevity, we do not address numeric implementation of the filtering methods. However, \ref{sec:noninteger_shift} is useful in this context. Second, we present a numeric inversion algorithm for bounded BRT data. Leveraging the rotational invariance of the two-dimensional Fourier transform, the directions $\theta_i, \theta_j$ are unconstrained in our algorithm. We also include regularization to address poor conditioning of the forward operator.

\subsection{Extending Cone Beam Transform Data}\label{sec:alg_cbt_mod}
\input{alg_cbt_mod}

\subsection{Inversion of BRT Data with Regularization}\label{sec:brt_inv}
\input{alg_brt_inv}

\section{Numerical Simulations}\label{sec:results}
\input{results}

\section{Discussion}\label{sec:discussion}
In conclusion we have demonstrated a new inversion algorithm for the BRT with improved performance with one scatter angle. Casting the BRT as a linear combination of CBTs provides insight on the minimum extent of sampling required for reconstruction and techniques to bound support of the data. Indexing the data by the scatter location, the BRT is an LSI operator. Analyzing the BRT as a linear operator, in the Fourier domain, yields a concise representation of its nullspace and highlights numerical sensitivities. This motivates the use of regularization which improves reconstruction from sampled data. Improving inversion of the BRT using one scatter angle supports extension to coherent scatter x-ray imaging problem which is angularly selective due to the momentum transfer.

\ack
\begin{itemize}
	\item J. A. O'Sullivan was supported by NIH R01 CA 212638.
	\item This paper has been improved substantially by the reviewers' comments.
\end{itemize}
\appendix

\section{Derivation of the Fourier Transform of the CBT}\label{sec:cbt_ft}
\input{cbt_Fourier}

\section{BRT Inversion by Fourier Analysis}\label{sec:brt_inverse_fourier}
\input{brt_inverse_fourier}

\section{Two-Dimensional Fourier Transform of a Parallelogram}\label{sec:parallelogram}
\input{parallelogram}

\section{Non-Integer Shifts of Sampled Signals}\label{sec:noninteger_shift}
\input{noninteger_shift}

\section*{References}

\bibliography{bibs}

\end{document}

%% file: introduction.tex
The broken ray transform (BRT) appears in the forward model of a number of imaging modalities and measurement geometries. It was first considered in the context of optical scatter imaging \cite{Florescu2009} and later applied to x-ray scatter imaging \cite{Katsevich2013}. The BRT occurs when the measured data are characterized by two rays sharing a common vertex. Time-of-flight positron emission tomography (TOF-PET) could also be considered in this framework (subject to the TOF ambiguity profile). The BRT has been considered for both translation-only measurement geometries~\cite{Florescu2009,Florescu2010,Florescu2011,Katsevich2013,Gouia2014,Zhao2014,Sherson2015,Ambartsoumian2019} and a rotate-shift measurement geometry \cite{Ambartsoumian2012,Sherson2015,Ambartsoumian2016}.

To our knowledge, all applications associated with the single-scatter BRT represent joint reconstruction problems. Two spatially-varying images must be resolved. For example, it may be necessary to recover attenuation despite nonuniform scatter density \cite{Florescu2010,Katsevich2013}, or recover two attenuation images at distinct energy levels \cite{Florescu2018}. Variations in the joint reconstruction problem have motivated several novel contributions related to the BRT. These contributions are not strictly academic. In application, the forward models must be tailored to the joint reconstruction problem. This may limit available or useful data. The joint reconstruction problem determines available BRT inversion strategies.
 
Before contrasting prior work, we first define a notional measurement geometry. We will use this to establish some notation and define a joint reconstruction problem involving the BRT. We generalize the joint reconstruction problem to cover coherent scatter x-ray imaging. This modality has received renewed interest recently, however joint reconstruction of scatter density and attenuation has not yet been addressed \cite{MacCabe2012,Brady2013,Odinaka2017}. In coherent scatter x-ray imaging, scatter density is highly sensitive to scatter angle. For this reason, we focus on BRT inversion using only two scatter angles. 

As a simplification we assume a mono-chromatic x-ray pencil-beam incident upon some media of interest. At point ${x\in\mathbb{R}^2}$ the beam interacts with the media and scatters coherently. We use  $\theta_i\in S^1$ to represent the direction of the source relative to the scatter location. The direction of the scattered photon is  ${\theta_j\in S^1}$. We assume it is detected by a columated detector. Due to the combination of a pencil beam and columated detector we assume the scatter location $x$ is known precisely. This measurement geometry is depicted in Figure \ref{fig:measurement_geometry}.

\begin{figure}%
	\centering
	\subfloat[Measurement gemoetry]{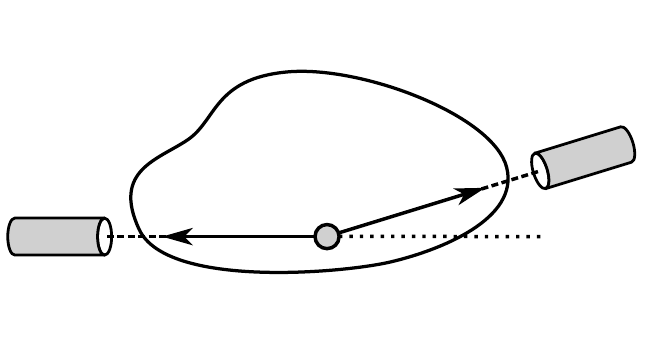\label{fig:measurement_geometry}}
	\subfloat[System model]{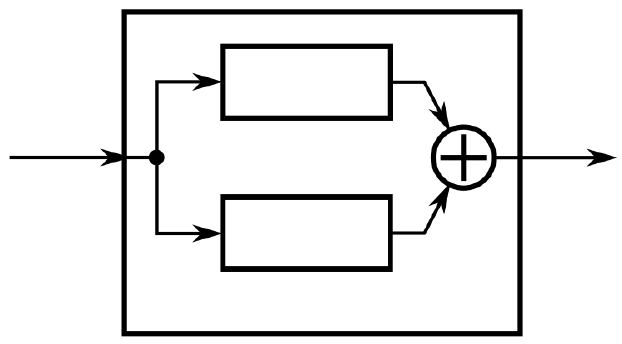\label{fig:lsi}}
	\caption{The generalized measurement geometry is depicted in Figure \ref{fig:measurement_geometry}. The location of the scattering event is indicated by $x$. The directions $\theta_i$, and $\theta_j$ represent the direction of the source and detector relative to $x$, respectively. Indexing the data by the scatter location, $x$, both the BRT and CBT are linear shift-invariant. In Figure \ref{fig:lsi} we depict the forward operator as a linear system. The system operates on an image, $\mu$, and returns data $g_{i,j}$.}
	\label{fig:geometryandsystem}
\end{figure}

The intensity measured at the detector largely depends on two media-specific images: attenuation and scatter density. The incident path is a straight line defined by $\theta_i,x$. The loss in intensity along this path due to attenuation is governed by Beer's law
\begin{equation}
\exp\left(-\int_0^\infty \mu(x - t\theta_i)dt\right).
\end{equation} 
We use $\mu(x): \mathbb{R}^2\rightarrow\mathbb{R}^+$ as the attenuation image representing both scatter and absorption. Intensity loss along the scatter path due to attenuation has a similar form and combines multiplicatively. For non-coherent scatter applications (e.g. fluorescence imaging) it may be necessary to distinguish the energy levels of the attenuation image before and after the scatter event. This has been investigated recently \cite{Florescu2018}.

Even in homogeneous media, the intensity observed at the detector may vary with respect to scatter angle (e.g. $\theta_i\cdot\theta_j$) and energy level. For coherent scatter imaging, the scatter density does not depend on these terms independently, but rather through Bragg's law \cite{Harding1987}. This relationship is summarized by the so-called momentum transfer ${q(s,E): (-1,1)\times \mathbb{R}^+\rightarrow \mathbb{R}^+}$
\begin{equation}
q(s,E) = 2\frac{E}{hc}\sqrt{\frac{1-s}{2}}.
\end{equation}
Here $h$ and $c$ are the Plank's constant and the speed of light, respectively. This definition is unconventional as we have chosen to define it over the cosine of the scatter angle, $s$, rather than the scatter angle directly. Scatter intensity for inhomogeneous media varies both spatially and with respect to momentum transfer. We use $f(x,q) : \mathbb{R}^2 \times \mathbb{R}^+\rightarrow\mathbb{R}^+$ to represent the scatter density image.

Combining the effects of attenuation and scatter density we arrive at the measurement function
\begin{equation}
\fl p(x,\theta_i,\theta_j,E) = f\left(x,q\left(-\theta_i\cdot\theta_j,E\right)\right)\exp\left(-\int_0^\infty \mu(x + t\theta_i) + \mu(x + t\theta_j)dt\right)\label{eq:pfull}. 
\end{equation}
In this expression we have omitted a number of terms necessary for accurate models of measured data. However, we assume the remaining terms are known multiplicative factors. Measured data can then be scaled to achieve this generalized form.

To simplify the notation we will make use of three common transforms. Borrowing the notation of Natterer \cite{Natterer2001}, we define the cone beam transform (CBT) $B$ of $\mu$
\begin{equation}
(B \mu)(x,\theta):=\int_0^\infty \mu(x + t\theta)dt. \label{eq:conebeamtransform}
\end{equation}
This transform appears in \eqref{eq:pfull}. In particular, the generalized model includes the linear combination of two cone beam transforms sharing a common vertex. This is commonly referred to as the broken ray transform 
\begin{equation}
(G \mu)(x,\theta_i,\theta_j):=(B \mu)(x,\theta_i) + (B \mu)(x,\theta_j). \label{eq:brokenraytransform}
\end{equation}
Denoting the left side of \eqref{eq:brokenraytransform} by data $g_{i,j}$, Figure \ref{fig:lsi} represents \eqref{eq:brokenraytransform} as a linear system with component CBT operators \eqref{eq:conebeamtransform}. 

In subsequent sections we will also make use of the  2D Radon transform
\begin{equation}
(R \mu)(v,\theta):=\int_{-\infty}^\infty \mu(v\theta^\perp + t\theta)dt. \label{eq:radon_restatement}
\end{equation}
Here $v\in\mathbb{R}^1$, and $\theta\in S^1$ represent the shift and rotate coordinates of the transform. We assume $\theta^\perp$ is uniquely defined by rotating $\theta$ counter-clockwise by $\pi/2$. 

Using these transforms we can express the log of the measured data
\numparts
\begin{eqnarray}
\fl\ln p(x,\theta_i,\theta_j,E) &= \ln f\left(x,q\left(\theta_i\cdot\theta_j,E\right)\right) - (B \mu)(x,\theta_i)-(B \mu)(x,\theta_j)\\
\fl&=\ln f\left(x,q\left(\theta_i\cdot\theta_j,E\right)\right) - (G \mu)(x,\theta_i,\theta_j).\label{eq:logdata} 
\end{eqnarray}
\endnumparts
The BRT is not directly available in \eqref{eq:logdata}. However, the term $f$ can be canceled with differential measurements \cite{Florescu2010} even for inhomogenous media. 

Given three scatter angles $\theta_i,\theta_j,\theta_k$ such that 
\begin{equation}
\theta_i\cdot\theta_k = \theta_j\cdot\theta_k\label{eq:condition_momentum},
\end{equation}
we have 
\begin{equation}
\fl \ln p(x,\theta_i,\theta_k,E) - \ln p(x,\theta_j,\theta_k,E) = -(B \mu)(x,\theta_i)+(B \mu)(x,\theta_j).\label{eq:signedbrokenraytransform}
\end{equation}
The condition \eqref{eq:condition_momentum} is only required when the scatter density is a function of momentum transfer. For some modalities, scatter density varies with respect to scatter angle according to a known function (e.g. Klein–Nishina). In such cases the data can be corrected and momentum transfer removed from \eqref{eq:logdata}. 

For clarification, we will refer to the right-hand side of \eqref{eq:signedbrokenraytransform} as the signed broken ray transform (SBRT) due to the sign change between CBTs. This is equivalent to the signed V-line transform \cite{Ambartsoumian2019}. Some authors have reserved their definition of the BRT for this later expression \cite{Katsevich2013}. While either definition of the BRT assumes a linear combination of two CBTs sharing a common vertex, the distinction is important for inversion. Also, while positive images yield positive BRT data, SBRT data may be negative.

For tomographic imaging applications it is common to index the data according to the source and detector locations. In contrast our indexing is somewhat unconventional. In the context of the BRT, Katsevich and Krylov were the first to demonstrate the benefits of indexing the data by the scatter location \cite{Katsevich2013}. Under this indexing schema, both the CBT and BRT are linear and shift-invariant (LSI). Linear systems analysis is therefore applicable to the CBT and BRT. Their relationship is depicted in Figure \ref{fig:lsi}. This is a central theme to our contribution as we are the first to consider the two-dimensional Fourier transform of the BRT. This perspective has benefits which we will demonstrate in subsequent sections.

Our focus is limited to 2D single-scatter imaging problems where scatter events are observed throughout the media of interest. This distinction is important because the terms broken ray transform and v-line transform have been used to describe a number of related problems. We distinguish BRT problems integrating over multiple reflections \cite{Hubenthal2014,deHoop2017} or integrating over multiple vertices \cite{Morvidone2010,Truong2011}. Some constrain the vertex locations along the perimeter of the measurement geometry \cite{Hubenthal2014,deHoop2017,Haltmeier2017}. This is generally motivated by the use of Compton cameras. In three dimensions this results in the cone transform \cite{Terzioglu2018}, which we distinguish from the cone beam transform \eqref{eq:conebeamtransform}\cite{Natterer2001} applicable to our measurement geometry.

The first analytic inversion formula for the BRT is due to Florescu et al. \cite{Florescu2011}. The global inversion formula requires only two scatter angles to recover the attenuation image in the presence of spatially varying scatter density. The inversion technique can be summarized as a three-step process. First, obtain the one-dimensional Fourier transform of the data. For the second step, each frequency is considered independently. Solve the resulting complex, one-dimensional, bounded differential equation. Third, obtain the inverse one-dimensional Fourier transform across the solutions. This yields an exact reconstruction of images with bounded support. The coordinate used to index the data in the original derivation were not linear-shift invariant. This was later derived using data indexed by the scatter location and generalized for higher dimensions \cite{Gouia2014}.

The number of available scatter angles is a discriminating factor in selecting a BRT inversion strategy. A local inversion formula was discovered by Katsevich and Krylov requiring 3 unique scatter angles \cite{Katsevich2013}. In contrast with prior results, \cite{Florescu2011}, their reconstructions demonstrated significant reduction in artifacts. This was later generalized for additional scatter angles and source locations \cite{Zhao2014}. While the attenuation map can be recovered locally, the recovery of the scatter density image still requires global reconstruction of the attenuation image. This, and the requirement \eqref{eq:condition_momentum} for coherent scatter imaging, motivates our interest in 2D BRT inversion techniques using only two scatter angles.

The initial results by Florescu et al. contained significant artifacts even for trivial phantoms \cite{Florescu2011}. These artifacts were broadly attributed to the nonlocal effects of integration. Artifacts in initial results exhibited striations at three distinct angles. Two of these angles are associated with the incident and scatter directions ($\theta_i$ and $\theta_j$). However, this did not directly address the third direction. This was later explored using micro-local analysis \cite{Sherson2015}. Sherson was the first to recognize $\theta_i + \theta_j$ as the direction of integration required for inversion.

Most recently a new inversion technique was developed by Ambartsoumian and Jebelli \cite{Ambartsoumian2019}. They used linear-shift invariant indexing of the data but did not employ the Fourier transform. They thoroughly and eloquently derive a new inversion technique by extending the Fundamental Theorem of Calculus to higher dimensions under a linear change of variables. They consider V-Line transformed (VLT) data defined by the linear combination CBTs along multiple directions $\{\theta_i\}$. Integrating VLT data along the direction $\sum_i \theta_i$ yields the integral of the image over the faceted cone defined by $\{\theta_i\}$ and the common vertex $x$. This unbounded volume can be reduced to a parallelepiped by linearly combining samples of this integral. Weighting the results, one obtains a reconstruction of the image averaged over this volume. This leads to a wonderfully concise inversion formula. In two dimensions they replace differentiation along the directions $\theta_i$ and $\theta_j$ with sample differences. This leaves only integration along the direction ${\theta_i + \theta_j}$. The consequence is potential blurring over the resulting parallelogram. This can be arbitrarily small for noise-free environments with high resolution data. For noisy data, the size of the parallelogram must be larger which effects blurring in the reconstruction. Additionally, artifacts appear along the direction of integration ($\theta_i + \theta_j$) \cite{Ambartsoumian2019}.

In the following we take a fresh look at the BRT as a linear shift-invariant operator. A linear systems perspective provides new insights on the transforms and tools for contrasting prior inversion formulas. More specifically, we demonstrate images with bounded support do not guarantee data with bounded support. We are the first to consider the minimum data required for reconstruction and techniques for bounding support of the data. The two-dimensional Fourier transform of the BRT operator is especially useful for contrasting inversion techniques. It exhibits zeros along direction of integration ($\theta_i + \theta_j$) \cite{Ambartsoumian2019}. The ensuing ambiguity can be resolved analytically with boundary conditions on the reconstructed image (e.g. 0 after subtracting the background level). However, this does not address numerical sensitivity. The poles in the forward operator, along the directions $\theta_i$ and $\theta_j$, also present numerical challenges. We contrast recent work \cite{Ambartsoumian2019} against prior inversion techniques \cite{Florescu2011,Gouia2014,Sherson2015} as different strategies for addressing the poles in the forward operator. To mitigate numerical issues associated with the forward operator, we advocate Tikhonov regularization in obtaining reconstructed images. This can be implemented efficiently in the Fourier domain for arbitrary angles $\theta_i$, $\theta_j$ due to rotational invariance of the two-dimensional Fourier transform.

%% file: diag_meas_geom.eps_tex
\begingroup%
  \makeatletter%
  \providecommand\color[2][]{%
    \errmessage{(Inkscape) Color is used for the text in Inkscape, but the package 'color.sty' is not loaded}%
    \renewcommand\color[2][]{}%
  }%
  \providecommand\transparent[1]{%
    \errmessage{(Inkscape) Transparency is used (non-zero) for the text in Inkscape, but the package 'transparent.sty' is not loaded}%
    \renewcommand\transparent[1]{}%
  }%
  \providecommand\rotatebox[2]{#2}%
  \newcommand*\fsize{\dimexpr\f@size pt\relax}%
  \newcommand*\lineheight[1]{\fontsize{\fsize}{#1\fsize}\selectfont}%
  \ifx\svgwidth\undefined%
    \setlength{\unitlength}{189.90101624bp}%
    \ifx\svgscale\undefined%
      \relax%
    \else%
      \setlength{\unitlength}{\unitlength * \real{\svgscale}}%
    \fi%
  \else%
    \setlength{\unitlength}{\svgwidth}%
  \fi%
  \global\let\svgwidth\undefined%
  \global\let\svgscale\undefined%
  \makeatother%
  \begin{picture}(1,0.52132422)%
    \lineheight{1}%
    \setlength\tabcolsep{0pt}%
    \put(0,0){\includegraphics[width=\unitlength]{diag_meas_geom.eps}}%
    \put(0.33351997,0.19121598){\color[rgb]{0,0,0}\makebox(0,0)[lt]{\lineheight{1.25}\smash{\begin{tabular}[t]{l}$\theta_i$\end{tabular}}}}%
    \put(0.09315091,0.21438926){\color[rgb]{0,0,0}\makebox(0,0)[t]{\lineheight{1.25}\smash{\begin{tabular}[t]{c}Source\end{tabular}}}}%
    \put(0.88113671,0.18813827){\color[rgb]{0,0,0}\makebox(0,0)[t]{\lineheight{1.25}\smash{\begin{tabular}[t]{c}Detector\end{tabular}}}}%
    \put(0.64361711,0.23527112){\color[rgb]{0,0,0}\makebox(0,0)[rt]{\lineheight{1.25}\smash{\begin{tabular}[t]{r}$\theta_j$\end{tabular}}}}%
    \put(0.48235879,0.20611403){\color[rgb]{0,0,0}\makebox(0,0)[t]{\lineheight{1.25}\smash{\begin{tabular}[t]{c}$x$\end{tabular}}}}%
    \put(0.43990508,0.32632085){\color[rgb]{0,0,0}\makebox(0,0)[t]{\lineheight{1.25}\smash{\begin{tabular}[t]{c}Media\end{tabular}}}}%
  \end{picture}%
\endgroup%

%% file: diag_lsi.eps_tex
\begingroup%
  \makeatletter%
  \providecommand\color[2][]{%
    \errmessage{(Inkscape) Color is used for the text in Inkscape, but the package 'color.sty' is not loaded}%
    \renewcommand\color[2][]{}%
  }%
  \providecommand\transparent[1]{%
    \errmessage{(Inkscape) Transparency is used (non-zero) for the text in Inkscape, but the package 'transparent.sty' is not loaded}%
    \renewcommand\transparent[1]{}%
  }%
  \providecommand\rotatebox[2]{#2}%
  \newcommand*\fsize{\dimexpr\f@size pt\relax}%
  \newcommand*\lineheight[1]{\fontsize{\fsize}{#1\fsize}\selectfont}%
  \ifx\svgwidth\undefined%
    \setlength{\unitlength}{178.63340759bp}%
    \ifx\svgscale\undefined%
      \relax%
    \else%
      \setlength{\unitlength}{\unitlength * \real{\svgscale}}%
    \fi%
  \else%
    \setlength{\unitlength}{\svgwidth}%
  \fi%
  \global\let\svgwidth\undefined%
  \global\let\svgscale\undefined%
  \makeatother%
  \begin{picture}(1,0.55420764)%
    \lineheight{1}%
    \setlength\tabcolsep{0pt}%
    \put(0,0){\includegraphics[width=\unitlength]{diag_lsi.eps}}%
    \put(0.66144952,0.05676832){\color[rgb]{0,0,0}\makebox(0,0)[t]{\lineheight{1.25}\smash{\begin{tabular}[t]{c}BRT, $\theta_i,\theta_j$\end{tabular}}}}%
    \put(0.4936922,0.40634784){\color[rgb]{0,0,0}\makebox(0,0)[t]{\lineheight{1.25}\smash{\begin{tabular}[t]{c}CBT, $\theta_i$\end{tabular}}}}%
    \put(0.4936922,0.1636196){\color[rgb]{0,0,0}\makebox(0,0)[t]{\lineheight{1.25}\smash{\begin{tabular}[t]{c}CBT, $\theta_j$\end{tabular}}}}%
    \put(0.10132265,0.34281902){\color[rgb]{0,0,0}\makebox(0,0)[t]{\lineheight{1.25}\smash{\begin{tabular}[t]{c}Image\end{tabular}}}}%
    \put(0.92894517,0.33210398){\color[rgb]{0,0,0}\makebox(0,0)[t]{\lineheight{1.25}\smash{\begin{tabular}[t]{c}Data\end{tabular}}}}%
    \put(0.101678,0.24020511){\color[rgb]{0,0,0}\makebox(0,0)[t]{\lineheight{1.25}\smash{\begin{tabular}[t]{c}$\mu$\end{tabular}}}}%
    \put(0.92908186,0.24189984){\color[rgb]{0,0,0}\makebox(0,0)[t]{\lineheight{1.25}\smash{\begin{tabular}[t]{c}$g_{i,j}$\end{tabular}}}}%
  \end{picture}%
\endgroup%

%% file: analysis_support.tex
It is helpful to distinguish segments of the boundary of $C$ with respect to the orthogonal basis $\theta$, $\theta^\perp$. For this we define the scalar values
\begin{eqnarray}
v^-_\theta &\coloneqq \min_{x\in C} x\cdot \theta^\perp \label{eq:v_}\\
v^+_\theta &\coloneqq \max_{x\in C} x\cdot \theta^\perp.\label{eq:vplus}\\
v_\theta &\coloneqq v_\theta^+-v_\theta^- \label{eq:v}
\end{eqnarray}
Additionally we define the auxiliary functions ${u^-_\theta,u^+_\theta:\left[v^-_\theta,v^+_\theta\right]\rightarrow\mathbb{R}}$
\begin{eqnarray}
u^-_\theta(v)\coloneqq \min t\textrm{, s.t. } t\theta + v\theta^\perp \in C\\
u^+_\theta(v)\coloneqq \max t\textrm{, s.t. } t\theta + v\theta^\perp \in C.\label{eq:u_}
\end{eqnarray}
We define non-overlapping line segments along the boundary of $C$ as functions, ${f^-,f^+:\left[v^-_\theta,v^+_\theta\right]\rightarrow C}$
\begin{eqnarray}
f^-_\theta(v)\coloneqq u^-_\theta(v)\theta + v\theta^\perp\\
f^+_\theta(v)\coloneqq u^+_\theta(v)\theta + v\theta^\perp.
\end{eqnarray}
Using these functions, we define the mutually exclusive regions ${C^{-}_\theta,C^{+}_\theta,V^{-}_\theta,V^{+}_\theta\subset\mathbb{R}^2}$:
\begin{eqnarray}
C^-_\theta&\coloneqq \left\{x : x\cdot\theta < u^-_\theta(x\cdot \theta^\perp;\theta), x\cdot\theta^\perp \in \left[v^-_\theta,v^+_\theta\right]\right\}\label{eq:C_}\\
C^+_\theta&\coloneqq \left\{x : x\cdot\theta > u^+_\theta(x\cdot \theta^\perp;\theta), x\cdot\theta^\perp \in \left[v^-_\theta,v^+_\theta\right]\right\}\label{eq:C+}\\
V^-_\theta&\coloneqq \left\{x :  x\cdot\theta^\perp < v^-_\theta\right\}\\
V^+_\theta&\coloneqq \left\{x :  x\cdot\theta^\perp > v^+_\theta\right\}.\label{eq:V+}
\end{eqnarray}  
These definitions are illustrated in Figure \ref{fig:cbt_data_regions}.

\begin{figure}%
	\centering
	\subfloat[CBT data regions]{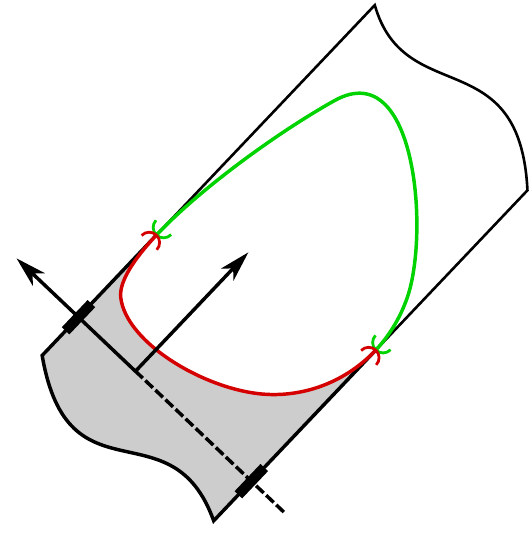\label{fig:cbt_data_regions}}\qquad
	\subfloat[BRT data regions]{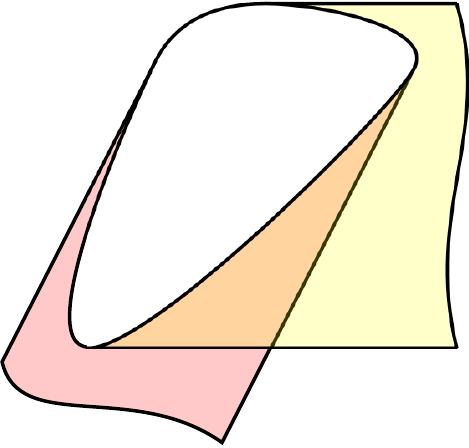\label{fig:brt_data_regions}}
	\caption{Images with bounded support, $C$, do not guarantee data with bounded support for either the CBT or BRT. For the CBT the support of the data is extended indefinitely in one direction, $-\theta$, over the region $C^-_\theta$ as depicted in Figure \ref{fig:cbt_data_regions}. For the BRT the support of the data is extended in two directions. Depending on the shape of $C$, the resulting regions $C_i^-$ and $C_j^-$ may intersect.}
	\label{fig:referenceImageAndBRT}
\end{figure}

With definitions in place, we make some observations regarding the support of the CBT data. We state them as three theorems. First, we limit support of the data. The CBT data are zero for all $x$ outside the support of the image, $C$, and the shadow region $C^-_\theta$. 
\begin{theorem}
$(B \mu_C)(x,\theta)=0$ for all $x\in \overline{C\cup C^-_\theta}$\label{th:cbt_0}
\end{theorem}
\begin{proof}
	The set $\overline{C\cup C^-_\theta}$ can be partitioned into three regions $C^+_\theta$, $V^+_\theta$, and $V^-_\theta$, where
	\begin{equation}
	\overline{C\cup C^-_\theta} = C^+_\theta \cup V^+_\theta \cup V^-_\theta.
	\end{equation}
Since $C^+_\theta\cap C=\varnothing$, we have $\mu_C(x)=0$ for all $x\in C^+_\theta$. Additionally, if $x\in C^+_\theta$, then according to the definition \eqref{eq:C+}, ${x+s \theta \in C^+_\theta}$ for all $s\geq0$. Therefore ${(B\mu_C)(x,\theta)=0}$ for all $x\in C^+_\theta$. The same is true for $V^-_\theta$ and $V^+_\theta$. 
\end{proof}

Next we observe the CBT, over the region $C_\theta^-$, is constant along the direction $\theta$. The values are determined by the Radon transform at $\theta$.

\begin{theorem}
	$(B \mu_C)(x,\theta)=(R \mu_C)(x\cdot\theta^\perp,\theta)$ for all $x\in C^-_\theta$\label{th:cbt_minus}
\end{theorem}
\begin{proof}
	For all $x\in C^-_\theta$, we can extend the integral of the CBT
	\begin{equation}
	(B\mu_C)(x,\theta) = \int_{-\infty}^{0}\mu_C(x + t\theta) + \int_{0}^{\infty}\mu_C(x + t\theta)dt
	\end{equation}
	where the first term is 0 due to the bounded support of $\mu_C$. Combining these integrals and expanding $x$ along the orthogonal basis vectors $\theta$ and $\theta^\perp$, we have
	\numparts
	\begin{eqnarray}
	(B\mu_C)(x,\theta) &= \int_{-\infty}^{\infty}\mu_C\left(\left(x\cdot\theta^\perp\right)\theta^\perp + t\theta\right)dt\\
	&= (R\mu_C)(x\cdot\theta^\perp,\theta).
	\end{eqnarray}	\endnumparts
\end{proof}

Finally, the Radon transform, for fixed direction $\theta$, is given by the CBT along the boundary of $C$.
\begin{theorem}
${(R \mu_C)(v,\theta)=(B\mu_C)(f^-_\theta(v),\theta)}$ for all ${v\in\left[v^-_\theta,v^+_\theta\right]}$ and ${(R \mu_C)(v,\theta)=0}$ for all  ${v\notin\left[v^-_\theta,v^+_\theta\right]}$.\label{th:cbt_bound}
\end{theorem}
\begin{proof}
For ${v\in\left[v^-_\theta,v^+_\theta\right]}$, we can expand
\numparts
\begin{eqnarray}
(R \mu_C)(v,\theta)&= \int_{-\infty}^{u^-(v)}\mu_C(v\theta^\perp + t\theta)dt+\int_{u^-(v)}^{\infty}\mu_C(v\theta^\perp + t\theta)dt\\
&=\int_{0}^{\infty}\mu_C(v\theta^\perp + u^-_\theta(v)\theta + t\theta)dt\\
&=(B\mu_C)(f^-_\theta(v),\theta)
\end{eqnarray}
\endnumparts
For $v<v^-$, we have ${\mu_C(v\theta^\perp+t\theta)=0}$, since ${v\theta^\perp+t\theta\in V^-_\theta}$ for all $t$. Therefore $(R \mu_C)(v,\theta)=0$ for $v<v^-_\theta$. The same can be shown for $v>v^+_\theta$.
\end{proof}

\begin{corollary}
	$(B \mu_C)(x,\theta)=(R \mu_C)(x\cdot\theta^\perp,\theta)$ for all ${x\in V^-_\theta \cup C^-_\theta \cup V^+_\theta}$\label{th:cbt_extend}
\end{corollary}
\begin{proof}
	Theorem \ref{th:cbt_minus} demonstrates equality for $x\in C^-_\theta$. For ${x\in V^-_\theta \cup V^+_\theta}$, we have $(B \mu_C)(x,\theta)=0$ according to Theorem \ref{th:cbt_0}. For ${x\in V^-_\theta}$, we have $x\cdot\theta^\perp<v^-_\theta$. Therefore, ${(R \mu_C)(x\cdot\theta,\theta)=0}$ according to Theorem \ref{th:cbt_bound}. Similarly, ${(R \mu_C)(x\cdot\theta,\theta)=0}$ for ${x\in V^+_\theta}$.
\end{proof}

Theorem \ref{th:cbt_minus} demonstrates images with bounded support do not guarantee CBT data with bounded support since $C^-_\theta$ is unbounded. This is problematic for discrete Fourier analysis. However, data outside the support of the original image is redundant. If $(B\mu_C)(x,\theta)$ is known for all $x\in C$ including its boundary, $(R \mu_C)(v,\theta)$ is available. Combining Theorem \ref{th:cbt_0} and Theorem \ref{th:cbt_minus}, the CBT is then known for all $x\in \mathbb{R}^2$. This is significant as there may be problems for which data are not available outside the support of the original image. This demonstrates samples along the boundary, or alternatively direct-path (ballistic) measurements, are sufficient. Once this minimum extent of data are available, CBT data can be extended arbitrarily. 

For our problems of interest CBT data are not available directly. The BRT is a linear combination of two CBTs sharing a common vertex. Similar to CBT data, bounded support of the image does not guarantee bounded support of the BRT data. The previous analysis of CBT data informs the sampling requirements on BRT data. Using the definition of the BRT \eqref{eq:brokenraytransform}, we distinguish two directions $\theta_i\neq\theta_j$. In addition to knowing $(G\mu_C)(x,\theta_i,\theta_j)$ for all $x\in C$, we additionally require the Radon transform in two direction: $(R \mu_C)(v,\theta_i)$, and $(R \mu_C)(v,\theta_j)$. The complication lies in the partitions of the BRT data. BRT data requires additional partitions which may overlap. Resolving the Radon transform, with respect to two directions, from the BRT data is more challenging.

Following the previous work, our definitions for $C$ and $\mu_C$ need not change. However, we use $i$ and $j$ to distinguish the directions the subscripts of the definitions \eqref{eq:v_}-\eqref{eq:V+}. These indices are used only in subscripts to avoid confusion with the imaginary unit $i\coloneqq \sqrt{-1}$. Depending on $C$, $\theta_i$, and $\theta_j$, the set $C_i^-\cap C_j^-$ may be nontrivial. The BRT data can be partitioned
\begin{equation}
\fl (G \mu_C)(x,\theta_i,\theta_j) = 
\cases{
	(B \mu_C)(x,\theta_i) + (B \mu_C)(x,\theta_j)  & for $x\in C$ \\
	(R \mu_C)(x\cdot\theta_i^\perp,\theta_i)+(R \mu_C)(x\cdot\theta_j^\perp,\theta_j)  & for $x\in C_i^-\cap C_j^-$ \\
	(R \mu_C)(x\cdot\theta_i^\perp,\theta_i)  & for $x\in C_i^- \setminus C_j^-$\\
	(R \mu_C)(x\cdot\theta_j^\perp,\theta_j)  & for $x\in C_j^- \setminus C_i^-$\\
	0 & otherwise.}\label{eq:brt_partitions}
\end{equation}
These regions are depicted in Figure \ref{fig:brt_data_regions}. 

In contrast to the CBT, we must distinguish $(R \mu_C)(v,\theta_i)$ from $(R \mu_C)(v,\theta_j)$. Over $C$ alone, they may not be directly available. We consider two scenarios. First, for some regions $C$ and scatter angles $\theta_i,\theta_j$, the set $C_i^-\cap C_j^-$ is empty. For example, this is true for rectangular $C$, when $\theta_i$ is parallel to a boundary of $C$, and $\theta_i\cdot\theta_j\leq0$. In such cases, $(R \mu_C)(v,\theta_i)$ from $(R \mu_C)(v,\theta_j)$ can be distinguished along the boundary of $C$. As a second scenario, forward scatter (ballistic) measurements at the two angles can be used to measure the Radon transforms directly. This would require new measurements. However, this may be useful for some modalities if measurements over the boundary of $C$ are not available.

The notation introduced in this section is also useful for simplifying the assumed support of the image. Due to shift-invariance of the BRT, we can assume the image is centered about the origin without loss of generality.
\begin{definition}
	Let $C$ represent a closed and bounded region in $\mathbb{R}^2$, and let $\theta_i$ and $\theta_j$ represent unique directions such that $|\theta_i\cdot\theta_j|<1$. We define $v_i^-$, $v_i^+$, $v_j^-$, $v_j^+$ using \eqref{eq:v_} and \eqref{eq:vplus}. Then, $C$ is centered with respect to $\theta_i$ and $\theta_j$ when both ${v_i^+ =-v_i^-}$ and ${v_j^+=-v_j^-}$.\label{def:centered}
\end{definition}

Parallelograms are an important geometric shape in the context of the BRT. This was first recognized by Ambartsoumian and Jebelli \cite{Ambartsoumian2019}. It is often convenient to extend $C$ to the circumscribed parallelogram.
\begin{definition}
	Let $C$ represent a closed and bounded region centered with respect to $\theta_i$ and $\theta_j$. The circumscribed parallelogram, with edges parallel to $\theta_i$ and $\theta_j$, is given by 
	\begin{equation}
	P \coloneqq \left\{x : \Pi_{v_i}\left(x\cdot\theta_i^\perp\right)\Pi_{v_j}\left(x\cdot\theta_j^\perp\right)>0\right\},\label{eq:p:v}
	\end{equation}
	where $v_i$, $v_j$ are defined according to \eqref{eq:v_}-\eqref{eq:v}. Clearly, $C\subseteq P \subset \mathbb{R}^2.$.\label{def:p}
\end{definition}
In \eqref{eq:p:v}, $P$ is expressed in terms of the orthogonal distance between parallel sides. Alternatively, we obtain the edge lengths $\alpha_i$ and $\alpha_j$ for the edges parallel to $\theta_i$ and $\theta_j$, respectively
\begin{eqnarray}
\alpha_i &\coloneqq v_j/\left|\det\left(\theta_i,\theta_j\right)\right|\label{eq:alphai}\\
\alpha_j &\coloneqq v_i/\left|\det\left(\theta_i,\theta_j\right)\right|\label{eq:alphaj}.
\end{eqnarray}
We can equivalently express $P$ in terms of the edge lengths
\begin{equation}
P = \left\{s_i\theta_i + s_j\theta_j; \left|s_i\right|\leq \alpha_i/2, \left|s_j\right|\leq \alpha_j/2\right\}.\label{eq:p:alpha}
\end{equation}
Related to $P$, we define a parallelogram indicator function in \ref{sec:parallelogram} and derive its two-dimensional Fourier transform. The results will be referenced frequently in subsequent sections.

%% file: diag_cbt_codomain.eps_tex
\begingroup%
  \makeatletter%
  \providecommand\color[2][]{%
    \errmessage{(Inkscape) Color is used for the text in Inkscape, but the package 'color.sty' is not loaded}%
    \renewcommand\color[2][]{}%
  }%
  \providecommand\transparent[1]{%
    \errmessage{(Inkscape) Transparency is used (non-zero) for the text in Inkscape, but the package 'transparent.sty' is not loaded}%
    \renewcommand\transparent[1]{}%
  }%
  \providecommand\rotatebox[2]{#2}%
  \newcommand*\fsize{\dimexpr\f@size pt\relax}%
  \newcommand*\lineheight[1]{\fontsize{\fsize}{#1\fsize}\selectfont}%
  \ifx\svgwidth\undefined%
    \setlength{\unitlength}{156bp}%
    \ifx\svgscale\undefined%
      \relax%
    \else%
      \setlength{\unitlength}{\unitlength * \real{\svgscale}}%
    \fi%
  \else%
    \setlength{\unitlength}{\svgwidth}%
  \fi%
  \global\let\svgwidth\undefined%
  \global\let\svgscale\undefined%
  \makeatother%
  \begin{picture}(1,1.02564103)%
    \lineheight{1}%
    \setlength\tabcolsep{0pt}%
    \put(0,0){\includegraphics[width=\unitlength]{diag_cbt_codomain.eps}}%
    \put(0.21846523,0.22707959){\color[rgb]{0,0,0}\makebox(0,0)[t]{\lineheight{1.25}\smash{\begin{tabular}[t]{c}$C^-_\theta$\end{tabular}}}}%
    \put(0.61447378,0.50947084){\color[rgb]{0,0,0}\makebox(0,0)[t]{\lineheight{1.25}\smash{\begin{tabular}[t]{c}$C$\end{tabular}}}}%
    \put(0.03256305,0.55675226){\color[rgb]{0,0,0}\makebox(0,0)[lt]{\lineheight{1.25}\smash{\begin{tabular}[t]{l}$\theta^\perp$\end{tabular}}}}%
    \put(0.45677336,0.5541414){\color[rgb]{0,0,0}\makebox(0,0)[t]{\lineheight{1.25}\smash{\begin{tabular}[t]{c}$\theta$\end{tabular}}}}%
    \put(0.15250598,0.48879106){\color[rgb]{0,0,0}\makebox(0,0)[t]{\lineheight{1.25}\smash{\begin{tabular}[t]{c}$v^+_\theta$\end{tabular}}}}%
    \put(0.52671936,0.12472865){\color[rgb]{0,0,0}\makebox(0,0)[lt]{\lineheight{1.25}\smash{\begin{tabular}[t]{l}$v^-_\theta$\end{tabular}}}}%
    \put(0.4898183,0.34471542){\color[rgb]{0,0,0}\makebox(0,0)[t]{\lineheight{1.25}\smash{\begin{tabular}[t]{c}$f^-(v;\theta)$\end{tabular}}}}%
    \put(0.61605384,0.71053757){\color[rgb]{0,0,0}\makebox(0,0)[t]{\lineheight{1.25}\smash{\begin{tabular}[t]{c}$f^+(v;\theta)$\end{tabular}}}}%
    \put(0.85796281,0.70196531){\color[rgb]{0,0,0}\makebox(0,0)[t]{\lineheight{1.25}\smash{\begin{tabular}[t]{c}$C^+_\theta$\end{tabular}}}}%
    \put(0.85670292,0.12300004){\color[rgb]{0,0,0}\makebox(0,0)[t]{\lineheight{1.25}\smash{\begin{tabular}[t]{c}$V^-_\theta$\end{tabular}}}}%
    \put(0.26319898,0.77474509){\color[rgb]{0,0,0}\makebox(0,0)[t]{\lineheight{1.25}\smash{\begin{tabular}[t]{c}$V^+_\theta$\end{tabular}}}}%
  \end{picture}%
\endgroup%

%% file: diag_brt_codomain.eps_tex
\begingroup%
  \makeatletter%
  \providecommand\color[2][]{%
    \errmessage{(Inkscape) Color is used for the text in Inkscape, but the package 'color.sty' is not loaded}%
    \renewcommand\color[2][]{}%
  }%
  \providecommand\transparent[1]{%
    \errmessage{(Inkscape) Transparency is used (non-zero) for the text in Inkscape, but the package 'transparent.sty' is not loaded}%
    \renewcommand\transparent[1]{}%
  }%
  \providecommand\rotatebox[2]{#2}%
  \newcommand*\fsize{\dimexpr\f@size pt\relax}%
  \newcommand*\lineheight[1]{\fontsize{\fsize}{#1\fsize}\selectfont}%
  \ifx\svgwidth\undefined%
    \setlength{\unitlength}{156bp}%
    \ifx\svgscale\undefined%
      \relax%
    \else%
      \setlength{\unitlength}{\unitlength * \real{\svgscale}}%
    \fi%
  \else%
    \setlength{\unitlength}{\svgwidth}%
  \fi%
  \global\let\svgwidth\undefined%
  \global\let\svgscale\undefined%
  \makeatother%
  \begin{picture}(1,0.9590682)%
    \lineheight{1}%
    \setlength\tabcolsep{0pt}%
    \put(0,0){\includegraphics[width=\unitlength]{diag_brt_codomain.eps}}%
    \put(0.47329335,0.63419812){\color[rgb]{0,0,0}\makebox(0,0)[t]{\lineheight{1.25}\smash{\begin{tabular}[t]{c}$C$\end{tabular}}}}%
    \put(0.80757818,0.37315541){\color[rgb]{0,0,0}\makebox(0,0)[t]{\lineheight{1.25}\smash{\begin{tabular}[t]{c}$C_i^-$\end{tabular}}}}%
    \put(0.37563705,0.12713683){\color[rgb]{0,0,0}\makebox(0,0)[lt]{\lineheight{1.25}\smash{\begin{tabular}[t]{l}$C_j^-$\end{tabular}}}}%
    \put(0.35122298,0.25484111){\color[rgb]{0,0,0}\makebox(0,0)[lt]{\lineheight{1.25}\smash{\begin{tabular}[t]{l}$C_i^-\cap C_j^-$\end{tabular}}}}%
  \end{picture}%
\endgroup%

%% file: analysis_modulation.tex
When the Fourier transform must be determined numerically, unbounded support of the BRT is problematic. Simply truncating BRT data effects blurring in the frequency domain. This corrupts the spectral representation and invalidates the previous Fourier reconstruction methods.  Alternatively, we consider convolving data in the spatial domain. We define a generalized point spread function (PSF) such that the shifted copies of the data combine destructively outside a bounded region of support. 

We consider the PSF
\begin{eqnarray}
m_{i,j}(x;a_i,a_j) &= \delta\left(x+\frac{a_i}{2}\theta_i +\frac{a_j}{2}\theta_j \right) - \delta\left(x-\frac{a_i}{2}\theta_i +\frac{a_j}{2}\theta_j \right)\nonumber \\
&\quad- \delta\left(x+\frac{a_i}{2}\theta_i -\frac{a_j}{2}\theta_j \right) + \delta\left(x-\frac{a_i}{2}\theta_i -\frac{a_j}{2}\theta_j \right).\label{eq:brt_mod_symmetric}
\end{eqnarray}
Here $a_i,a_j > 0$, determine the shift lengths. The expression \eqref{eq:brt_mod_symmetric} has the Fourier transform
\begin{equation}
\doublehat{m}_{i,j}(w;a_i,a_j) = -4\sin\left(\pi a_i w\cdot\theta_i\right)\sin\left(\pi a_j w\cdot\theta_j\right).\label{eq:brt_mod_symmetric_ft}
\end{equation}
To reduce the number of variables defined we introduce new notation to distinguish signals, which support expansion using the PSF function \eqref{eq:brt_mod_symmetric}. We define
\begin{eqnarray}
\doublehat{g}^{m}_{i,j}(w;a_i,a_j) \coloneqq \doublehat{g}_{i,j}(w)\doublehat{m}_\theta(w;a_i,a_j) \label{eq:brt_modulated_fourier_1}\\
\doublehat{\mu}^{m}_{i,j}(w;a_i,a_j) \coloneqq \doublehat{\mu}_{C}(w)\doublehat{m}_\theta(w;a_i,a_j). \label{eq:img_modulated_fourier_1}
\end{eqnarray}
The same superscript $m$ will be subsequently applied to continuous signals in the spatial domain, and sampled signals. Plugging \eqref{eq:brt_ft} and \eqref{eq:brt_mod_symmetric} into \eqref{eq:brt_modulated_fourier_1} we have
\begin{eqnarray}
\fl \doublehat{g}^{m}_{i,j}(w;a_i,a_j) =  \doublehat{\mu}_C(w) &\left[-4\sin\left(\pi a_i w\cdot\theta_i\right)\sin\left(\pi a_j w\cdot\theta_j\right) \frac{-w\cdot\left(\theta_i+\theta_j\right)}{i2\pi \left(w\cdot\theta_i\right)\left(w\cdot\theta_j\right)}\right.\nonumber\\
\fl &-4\sin\left(\pi a_i w\cdot\theta_i\right)\sin\left(\pi a_j w\cdot\theta_j\right)\frac{1}{2}\delta\left(w\cdot\theta_i\right)\nonumber\\
\fl &\left.-4\sin\left(\pi a_i w\cdot\theta_i\right)\sin\left(\pi a_j w\cdot\theta_j\right)\frac{1}{2}\delta\left(w\cdot\theta_j\right)\right].
\end{eqnarray}
The inverse two-dimensional Fourier transform of this expression involves integration over $w$. Due to the sampling property of the delta function, and since $\sin(0)=0$, the final two bracketed terms vanish under integration. By the uniqueness of the Fourier transform, we have
\numparts
\begin{eqnarray}
\doublehat{g}^{m}_{i,j}(w;a_i,a_j) &= \doublehat{\mu}_C(w)\doublehat{m}_{i,j}(w;a_i,a_j)\frac{-w\cdot\left(\theta_i+\theta_j\right)}{i2\pi \left(w\cdot\theta_i\right)\left(w\cdot\theta_j\right)}\label{eq:brt_modulated_fourier}\\
&= \doublehat{\mu}_{i,j}^m\left(w;a_i,a_j\right) \doublehat{h}_{i,h}(w).\label{eq:brt_filt_forward}
\end{eqnarray}
\endnumparts
In \eqref{eq:brt_filt_forward} we make use of both \eqref{eq:img_modulated_fourier_1}, and \eqref{eq:h}. Since the BRT is LSI, this result is expected. Filtering the input to an LSI system is equivalent to filtering the output. The significance is that the delta functions vanish when we filter the data using \eqref{eq:brt_mod_symmetric}.

We obtain another useful form by expanding \eqref{eq:brt_modulated_fourier} using \eqref{eq:brt_mod_symmetric_ft}. We reappropriate the denominator of $\doublehat{h}_{i,j}(w)$ to find
\begin{equation}
\fl \doublehat{g}^{m}_{i,j}(w;a_i,a_j)  = -i 2\pi w\cdot\left(\theta_i+\theta_j\right) \doublehat{\mu}_C(w)a_i a_j \sinc\left(a_i w\cdot\theta_i\right)\sinc\left(a_j w\cdot\theta_j\right)\label{eq:brt_sinc}.
\end{equation}
The product of $\sinc$ functions in \eqref{eq:brt_sinc} is associated with a parallelogram window function as demonstrated in \ref{sec:parallelogram}. This motivates the definition
\begin{equation}
\doublehat{\mu}_{i,j}^p\left(w;a_i,a_j\right)\coloneqq \frac{\doublehat{\mu}_C(w) \doublehat{p}_{i,j}\left(w;a_i,a_j\right)}{a_i a_j\left|\det\left(\theta_i,\theta_j\right)\right|}\label{eq:img_modulated_fourier_p}
\end{equation}
where $\doublehat{p}_{i,j}\left(w;a_i,a_j\right)$ is defined according to \eqref{eq:pft:a}. The scaling is motivated by \eqref{eq:p:area}. Using \eqref{eq:img_modulated_fourier_p} in \eqref{eq:brt_sinc}, we have
\begin{equation}
\doublehat{g}^{m}_{i,j}(w;a_i,a_j) =-i2\pi w\cdot\left(\theta_i+\theta_j\right)a_i a_j \doublehat{\mu}_{i,j}^p\left(w;a_i,a_j\right)\label{eq:brt_p}.
\end{equation}
Taking the inverse two-dimensional Fourier transform of \eqref{eq:brt_p} we find 
\begin{equation}
g^{m}_{i,j}(x;a_i,a_j) = -\frac{d}{d \left(\theta_i+\theta_j\right)}a_i a_j \mu_{i,j}^p\left(x;a_i,a_j\right).\label{eq:brt_mod_data}
\end{equation}
Here the first term represents the directional derivative in the direction $\theta_i+\theta_j$. This is clearly not a unit vector. In this form we observe $g^{m}_{i,j}(x;a_i,a_j)$ has bounded support.

\begin{theorem}
	For an absolutely integrable image with bounded support, filtering the BRT data with the PSF \eqref{eq:brt_mod_symmetric} bounds support of the data for all $a_i,\,a_j\in\left(0,\infty\right)$. Additionally, the data are finite everywhere.
\end{theorem}
\begin{proof}
	Without loss of generality, we assume the support of the image $\mu_C(x)$ is bounded by the circumscribed parallelogram, $P$, according to Definition \ref{def:p}.	We first observe $\mu_{i,j}^p\left(x;a_i,a_j\right)$ has bounded support. We define
	\numparts
	\begin{eqnarray}
	f(x) &\coloneqq \mu_{i,j}^p\left(x;a_i,a_j\right)a_i a_j\left|\det\left(\theta_i,\theta_j\right)\right|\\
	&= \mu_C(x)*p_{i,j}\left(x;a_i,a_j\right)\label{eq:muconvp}.
	\end{eqnarray}
	\endnumparts
	The indicator function $p_{i,j}\left(x;a_i,a_j\right)$, defined by \eqref{eq:p}, has bounded support over a parallelogram similar to $P$ in \eqref{eq:p:v}. Taking the convolution of two functions defined over similar parallelograms, the support of the result is also bounded by a similar parallelogram. This limits support of $f(x)$ to a parallelogram with sides parallel to $\theta_i$ and $\theta_j$ with perpendicular distances $v_i+b_i$ and $v_j+b_j$, respectively. The variables $b_i$ and $b_j$ are related to $a_j$ and $a_i$ according to \eqref{eq:a2vi} and \eqref{eq:a2vj}, respectively.

	Using \eqref{eq:muconvp} in \eqref{eq:brt_mod_data}, we have
	\begin{equation}
	|g^{m}_{i,j}(x;a_i,a_j) | = \frac{1}{|\det\left(\theta_i,\theta_j\right)|} \left|\frac{d}{d \left(\theta_i+\theta_j\right)}f(x) \right|
	\end{equation}
	Outside the the region of support of $f(x)$, its directional derivative is also zero. Therefore $g^{m}_{i,j}(x;a_i,a_j)$ has bounded support.
	
	To show $g^{m}_{i,j}(x;a_i,a_j)$ is finite everywhere, we consider
	\numparts 
	\begin{eqnarray}
	\fl\left|g^{m}_{i,j}(x;a_i,a_j)\right| &= \left|g_{i,j}(x) * m_{i,j}(x;a_i,a_j)\right|\\
	\fl&= \left|g_{i,j}\left(x+\frac{a_i}{2}\theta_i +\frac{a_j}{2}\theta_j \right) - g_{i,j}\left(x-\frac{a_i}{2}\theta_i +\frac{a_j}{2}\theta_j \right)\right. \nonumber\\
	\fl&\phantom{=|}\,\left.- g_{i,j}\left(x+\frac{a_i}{2}\theta_i -\frac{a_j}{2}\theta_j \right) + g_{i,j}\left(x-\frac{a_i}{2}\theta_i -\frac{a_j}{2}\theta_j \right)\right|\\
	\fl&\leq 4 \|g_{i,j}(x)\|_\infty.
	\end{eqnarray}
	\endnumparts
	This is finite due to the assumption $\mu_C(x)$ is integrable.
\end{proof}

%% file: analysis_reconstruct.tex
Bounded BRT data facilitates numeric inversion in the frequency domain. We consider two inversion strategies. The size of the available spreading parameters $a_i$ and $a_j$ in \eqref{eq:brt_mod_symmetric} plays an important role in selecting an inversion strategy. In both cases we reconstruct a version of the desired image subject to convolution. However, the PSFs associated with the reconstructed images are different.

Multiplying both sides of \eqref{eq:brt_filt_forward} by the inverse of \eqref{eq:h} we have the relationship
\begin{equation}
\doublehat{\mu}^m_{i,j}\left(w;a_i,a_j\right) = \doublehat{g}^{m}_{i,j}(w;a_i,a_j)\frac{-i2\pi \left(w\cdot\theta_i\right)\left(w\cdot\theta_j\right)}{ w\cdot\left(\theta_i+\theta_j\right)},\quad\forall\,w\notin \Theta_{i,j}.\label{eq:mumft}
\end{equation}
This is similar to \eqref{eq:brt_inv_fourier}. However, the reconstruction is subject to multiplication with the PSF \eqref{eq:brt_mod_symmetric_ft}. Analytically, we can recover $\doublehat{\mu}^m_{i,j}\left(w;a_i,a_j\right)$ from $\doublehat{g}^{m}_{i,j}(w;a_i,a_j)$ using \eqref{eq:mumft} and continuity assumptions or, equivalently, boundary conditions on $\mu^m_C(x)$. 

We find a representation of the left hand side of \eqref{eq:mumft} in the spatial domain by taking the inverse two-dimensional Fourier transform of \eqref{eq:img_modulated_fourier_1}
\begin{eqnarray}
\fl\mu^m_{i,j}(x;a_i,a_j) &= \mu_C\left(x+\frac{a_i}{2}\theta_i +\frac{a_j}{2}\theta_j \right) - \mu_C\left(x-\frac{a_i}{2}\theta_i +\frac{a_j}{2}\theta_j \right) \nonumber\\
\fl&\phantom{=}\,-\mu_C\left(x+\frac{a_i}{2}\theta_i -\frac{a_j}{2}\theta_j \right) + \mu_C\left(x-\frac{a_i}{2}\theta_i -\frac{a_j}{2}\theta_j \right).\label{eq:mum}
\end{eqnarray}
For small $a_i$ and $a_j$, the image copies will overlap. As $a_i$ and $a_j$ increase we can reconstruct $\mu_C\left(x\right)$ from segments without overlap. 
\begin{theorem}
	An image with with bounded support, $\mu_C(x)$, can be recovered from filtered BRT data $g^m_{i,j}(x; a_i, a_j)$ when ${a_i>v_j/\left|2\det\left(\theta_i,\theta_j\right)\right|}$ and ${a_j>v_i/\left|2\det\left(\theta_i,\theta_j\right)\right|}$ for $v_i$, $v_j$ defined according to \eqref{eq:v}.
\end{theorem}
\begin{proof}
A portion of the image $\mu_C(x)$, without overlap, is associated with each shifted copy in \eqref{eq:mum}. When the shifts are sufficiently large, the partial images can be combined to reconstruct the original image. To demonstrate this it is useful to first extend $C$ to the circumscribed parallelogram $P$ in \eqref{eq:p:alpha}. The edge lengths of this geometric region bound the minimum shift lenghts for image recovery.

To emphasize $P$ as the assumed region of support, we use $\mu_P(x)$. Since $C\subseteq P\subset\mathbb{R}^2$, we have $\mu_P(x)=\mu_C(x)$, for all $x\in\mathbb{R}^2$. For $a_i>\alpha_i/2$ and $a_j>\alpha_j/2$ we can recover $\mu_P(x)$ from $\mu^m_{i,j}(x)$ using
\begin{equation}
\fl \mu_P\left(s_i\theta_i + s_j\theta_j\right) = \cases{
\mu^m_{i,j}\left(\left(s_i-\frac{a_i}{2}\right)\theta_i + \left(s_j-\frac{a_j}{2}\right)\theta_j\right), & $s_i,s_j\leq 0$\\
-\mu^m_{i,j}\left(\left(s_i-\frac{a_i}{2}\right)\theta_i + \left(s_j+\frac{a_j}{2}\right)\theta_j\right), & $s_i\leq 0, s_j \geq 0$\\
-\mu^m_{i,j}\left(\left(s_i+\frac{a_i}{2}\right)\theta_i + \left(s_j-\frac{a_j}{2}\right)\theta_j\right), & $s_i\geq 0, s_j\leq 0$\\
\mu^m_{i,j}\left(\left(s_i+\frac{a_i}{2}\right)\theta_i + \left(s_j+\frac{a_j}{2}\right)\theta_j\right), & $s_i,s_j\geq 0$.}
\end{equation}
Each case can be expanded as a series of four terms using \eqref{eq:mum}. However, three of these terms are zero due to the support of $P$ in \eqref{eq:p:alpha}. Combining the four cases we recover $\mu_P(x)$, and therefore $\mu_C(x)$, for all $x\in\mathbb{R}^2$. Expanding $a_i>\alpha_i/2$ and $a_j>\alpha_j/2$ using \eqref{eq:alphai}, \eqref{eq:alphaj} we obtain the boundary in the stated form.
\end{proof}

In general this approach requires large data sets to obtain $\doublehat{g}^{m}_{i,j}(w;a_i,a_j)$. For many cases we can extend BRT data using the techniques in Section \ref{sec:datasupport}. However, when we are limited to small $a_i$ and $a_j$, another approach is necessary. 

Alternatively, we can simply recover $ \doublehat{\mu}_{i,j}^p\left(w;a_i,a_j\right)$. From \eqref{eq:brt_p}, we have
\begin{equation}
 \doublehat{\mu}_{i,j}^p\left(w;a_i,a_j\right)= \frac{-\doublehat{g}^{m}_{i,j}(w;a_i,a_j)}{i2\pi a_i a_j w\cdot\left(\theta_i+\theta_j\right)},\quad\forall\,w\notin \Theta_{i,j}.\label{eq:mumpft}
\end{equation}
Taking the inverse two-dimensional Fourier transform of \eqref{eq:mumpft} we have
\begin{eqnarray}
\mu_{i,j}^p\left(x;a_i,a_j\right) = \frac{-1}{\alpha_i\alpha_j\|\theta_i+\theta_j\|}\int_{-\infty}^0 g^m_{i,j}\left(x - t\frac{\theta_i+\theta_j}{\|\theta_i+\theta_j \| }\right)dt.\label{eq:ambartinverse}
\end{eqnarray}
This is equivalent to the inversion formula of Ambartsoumian and Jebelli \cite{Ambartsoumian2019}. For sampled data this formula can be implemented easily whenever the direction of integration is aligned with a sampling axis. For other cases, the frequency domain representation \eqref{eq:mumpft} is useful. 

We emphasize $\mu_{i,j}^p\left(x;a_i,a_j\right)\neq \mu_{C}\left(x\right)$. Taking the inverse two-dimensional Fourier transform of \eqref{eq:img_modulated_fourier_p}, we have
\begin{equation}
 \mu_{i,j}^p\left(x;a_i,a_j\right) = \frac{\mu(x) * p_{i,j}(x;a_i,a_j)}{a_i a_j\left|\det\left(\theta_i,\theta_j\right)\right|}.
\end{equation}
This demonstrates the recovered image as a blurring of the original image with a parallelogram window function. For high resolution, noise-free, data the size of this window can be made arbitrarily small. The recovery \eqref{eq:ambartinverse} is only approaches $\mu_C(x)$ in a limiting sense \cite{Ambartsoumian2019}.

%% file: alg_cbt_mod.tex
We consider CBT data sampled uniformly over a rectangular region. For consistency with previous definitions, we expand $x$ along two scalar axes $x = (t,y)$. For the two axes we use subscripts to distinguish the number of samples $N_t$, $N_y$ and the sample spacing $\Delta_t$, $\Delta_y$. We collect the available data in the $N_y\times N_x$ matrix $B$. The elements represent samples of the CBT data
\begin{equation}
\fl [B]_{n,m} = \left(B \mu_C\right)((t_B + (m-1)\Delta_t, y_B + (n-1)\Delta_y), (\cos\xi,\sin\xi))
\end{equation}
for $n\in\{1,\ldots,N_y \}$, and $m\in\{1,\ldots,N_t\}$. We expand the direction $\theta=(\cos\xi,\sin\xi)$. The spatial location associated with sample $B_{1,1}$ is $x_B = (t_B, y_B)$. In this configuration the $y$ coordinate increases with the row index $n$, and the $t$ coordinate increases with the column index $m$. It is not necessary to distinguish the terms $\Delta_t$, $\Delta_y$, and $\xi$ for most of the computations related to sampled CBT data. For convenience we define
\begin{equation}
\lambda \coloneqq \frac{\Delta_t}{\Delta_y}\tan\xi, \label{eq:alpha}
\end{equation}
which is a sufficient input for algorithms on uniformly sampled data.

Extending CBT data in the direction $\theta$ is trivial. For this we need only consider ${x\in V^-_\theta\cup C^+_\theta \cup V^+_\theta}$,  where $(B \mu_C)(x,\theta)=0$ according to Theorem \ref{th:cbt_0}. Zero padding is sufficient.

Extending the data in the direction $-\theta$ is nontrivial. For simplicity we first consider only $\xi\in(0,\pi/2)$. Figure \ref{fig:cbt_extend_diagram} illustrates the problem of extending the data, $B$, into the quadrants $Q2$, $Q3$, and $Q4$. The Radon transform serves as a proxy for extending the data according to Corollary \ref{th:cbt_extend}. We assume the first row and column comprise no samples interior to $C$ such that these data are samples of ${(R\mu_C)(x\cdot\theta^\perp,\theta)}$. We can then extend the data using $(B \mu_C)(x,\theta)=(R \mu_C)(x\cdot\theta^\perp,\theta)$. A brute-force approach would be to resample the Radon transform for each new data point. A computationally efficient approach is to extend the CBT data by shifting samples along the boundaries. This process is detailed in Algorithm \ref{alg:cbt_extend}. For $\xi\notin(0,\pi/2)$, we can still use Algorithm \ref{alg:cbt_extend} by suitably flipping the inputs and outputs.

\begin{figure}
	\centering{
		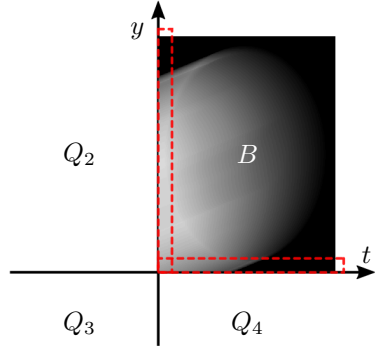
		\caption{To extend the sampled CBT data, $B$, we use only the first row and first column as indicated by the dashed boxes. We first extend the data in the direction $-t$ to synthesize data in the second quadrant (indicated by $Q_2$). We then extend the combined data in the direction $-y$. This synthesizes data in the third and fourth quadrants (indicated by $Q_3$ and $Q_4$). The process is detailed in Algorithm \ref{alg:cbt_extend}.}
		\label{fig:cbt_extend_diagram}
	}
\end{figure}

\begin{algorithm}
	\caption{\textsc{CbtExtend}: Extend CBT data from a rectangular, uniformly sampled region. We assume the direction of integration is positive. Interpreting the available data as occupying the first quadrant, we extend the data into quadrants 2-4 as depicted in Figure \ref{fig:cbt_extend_diagram}. The inputs ${\bf b}_y$ and ${\bf b}_t$ represent the first column and row of the data matrix $B$, respectively. Input $\lambda$ is given by \eqref{eq:alpha}. The inputs $M_t$ and $M_y$ indicate the number of requested samples in the direction $-t$ and $-y$, respectively. The input $p$ indicates desired padding when using Algorithm \ref{alg:nonintshift} presented in \ref{sec:noninteger_shift}. We use $\vertcat$ and $\flipud$ to vertically concatenate and vertically flip matrices, respectively.}
	\label{alg:cbt_extend}
	\begin{algorithmic}[1]
		\Require ${\bf b}_y\in \mathbb{R}^{N_y}$, ${\bf b}_t\in \mathbb{R}^{N_t}$, $ \lambda\in \mathbb{R}^+$, $M_t, M_y, p\in \mathbb{Z}^+$
		\Ensure $Q_2\in\mathbb{R}^{N_y\times M_t}$, $Q_3\in\mathbb{R}^{M_y\times M_t}$, $Q_4\in\mathbb{R}^{M_y\times N_t}$
		\State ${\bf x}_R = {\bf b}_y$ \Comment{Expand $B$ in direction $-t$}
		\State ${\bf x}_L = \flipud({\bf b}_t(2:1+p))$
		\State ${\bf s} =  \lambda\left[\matrix{-M_t & -M_t+1 & \cdots & -1}\right]^T$
		\State $p_W= p + M_y$
		\State $W = \textsc{NonIntShift}({\bf x}_R,\, {\bf s},\, p_W,\, {\bf x}_L)$\Comment{Algorithm \ref{alg:nonintshift}}
		\State $Q_2 = W(1:N_y,:)$
		
		\State ${\bf x}_R = \vertcat(Q_2(1,:)^T,{\bf b}_t)$ 	\Comment{Expand $[Q_2\, B]$ in direction $-y$ }
		\State ${\bf x}_L = \flipud(Q_2(2:1+p,1))$
		\State ${\bf s} =  \lambda^{-1}\left[\matrix{-M_y& -M_y+1 & \cdots & -1}\right]^T$
		\State $p_W = p + \lceil  \lambda^{-1} M_y \rceil$
		\State $W = \textsc{NonIntShift}({\bf x}_R,\, {\bf s},\, p_W,\, {\bf x}_L)$
		\State $Q_3 = W(1:M_t,:)^T$
		\State $Q_4 = W(M_t+1:N_t+M_t,:)^T$
	\end{algorithmic}
\end{algorithm}

We can adapt this process to an important class of BRT problems. We consider the incident direction aligned with the $t$-axis and $\theta_i\cdot\theta_j<0$. Specifically, we use $\theta_i = (-1,0)$ and expand $\theta_j=(\cos\xi,\sin\xi)$. We construct the BRT data matrix $G\in\mathbb{R}^{N_y\times N_t}$ with elements
\begin{equation}
\fl [G]_{n,m} = \left(G \mu_C\right)((t_G + (m-1)\Delta_t, y_G + (n-1)\Delta_y), (-1,0), (\cos\xi,\sin\xi)).\label{eq:g:sampled}
\end{equation}
The previous definitions for $\Delta_t$, $\Delta_y$, and $\lambda$ remain applicable.

Extending BRT data requires knowledge of both $(R \mu_C)(v,\theta_i)$ and $(R \mu_C)(v,\theta_j)$. We assume BRT data are sampled beyond the support of the image, such that no boundary samples of $G$ correspond to points within $C$. For $\theta_i = (-1,0)$, and ${|\xi|<\pi/2}$ this implies ${[G]_{1,N_t} = [G]_{N_y,N_t} = 0}$. In this case $(R \mu_C)(v,\theta_i)$ can be recovered from the last column of $G$. We can extend BRT data in the direction $-\theta_i$ simply by repeating the last column. For $\xi>0$, the last row (maximum $y$) of $G$ is then zero. The function $(R \mu_C)(v,\theta_j)$ can be recovered from the first column of $G$ and the first row. The BRT data can be extended in the direction $-\theta_j$ using Algorithm \ref{alg:cbt_extend}. Alternatively, for $\xi<0$,  $(R \mu_C)(v,\theta_j)$ can be recovered from the first column of $G$ and the last row. The BRT data can still be extended in the direction $-\theta_j$ using Algorithm \ref{alg:cbt_extend}. However, the inputs and outputs must be flipped accordingly.

%% file: diag_cbt_extension.eps_tex
\begingroup%
  \makeatletter%
  \providecommand\color[2][]{%
    \errmessage{(Inkscape) Color is used for the text in Inkscape, but the package 'color.sty' is not loaded}%
    \renewcommand\color[2][]{}%
  }%
  \providecommand\transparent[1]{%
    \errmessage{(Inkscape) Transparency is used (non-zero) for the text in Inkscape, but the package 'transparent.sty' is not loaded}%
    \renewcommand\transparent[1]{}%
  }%
  \providecommand\rotatebox[2]{#2}%
  \ifx\svgwidth\undefined%
    \setlength{\unitlength}{154.26028442bp}%
    \ifx\svgscale\undefined%
      \relax%
    \else%
      \setlength{\unitlength}{\unitlength * \real{\svgscale}}%
    \fi%
  \else%
    \setlength{\unitlength}{\svgwidth}%
  \fi%
  \global\let\svgwidth\undefined%
  \global\let\svgscale\undefined%
  \makeatother%
  \begin{picture}(1,0.88265305)%
    \put(0,0){\includegraphics[width=\unitlength]{diag_cbt_extension.eps}}%
    \put(0.18846295,0.47403099){\color[rgb]{0,0,0}\makebox(0,0)[b]{\smash{$Q_2$}}}%
    \put(0.18846299,0.0630384){\color[rgb]{0,0,0}\makebox(0,0)[b]{\smash{$Q_3$}}}%
    \put(0.59863607,0.0630384){\color[rgb]{0,0,0}\makebox(0,0)[b]{\smash{$Q_4$}}}%
    \put(0.59863607,0.46931469){\color[rgb]{0,0,0}\makebox(0,0)[b]{\smash{${\color{white}B}$}}}%
    \put(0.87796473,0.22520047){\color[rgb]{0,0,0}\makebox(0,0)[lb]{\smash{$t$}}}%
    \put(0.34616592,0.78654367){\color[rgb]{0,0,0}\makebox(0,0)[rb]{\smash{$y$}}}%
  \end{picture}%
\endgroup%

%% file: alg_brt_inv.tex
Filtering ensures bounded support of $g^m_{i,j}(x;a_i,a_j)$. However, recovery of $\doublehat{\mu}^m_{i,j}(w;a_i,a_j)$ is still ill-posed due to conditioning of $\doublehat{h}_{i,j}(w)$. For this we use Tikhonov regularization which can be applied sample-wise in the frequency domain.

We restate \eqref{eq:h} as an expression of scalar values by expanding $w = (w_t,w_y)$, $\theta_i = (\cos\xi_i, \sin\xi_i)$, and $\theta_j = (\cos\xi_j, \sin\xi_j)$ 
\begin{equation}
\fl \doublehat{h}_{i,j}((w_t,w_y)) = \frac{-w_t\left(\cos\xi_i+\cos\xi_j\right)-w_y\left(\sin\xi_i+\sin\xi_j\right)}{i2\pi \left(w_t\cos\xi_i+w_y\sin\xi_i\right)\left(w_t\cos\xi_j+w_y\sin\xi_j\right)}.\label{eq:doublehathscalar}
\end{equation}
Notice this expression is commutative with respect to $\xi_i$ and $\xi_j$. We define the system matrix $\doublehat{H}$ by sampling \eqref{eq:doublehathscalar} uniformly
\begin{equation}
[\doublehat{H}]_{n,m} = \doublehat{h}_{i,j}\left(\frac{m}{N_t\Delta_t},\frac{n}{N_y\Delta_y}\right).\label{eq:Hsampled}
\end{equation}
The discrete analog of \eqref{eq:brt_filt_forward} is then
\begin{equation}
\doublehat{G}^m = \doublehat{\Psi}^m \odot \doublehat{H}\label{eq:discfwd}.
\end{equation}
Here we have used $\doublehat{G}^m$ to represent the two-dimensional discrete Fourier transform of $G^m$, the filtered analog of \eqref{eq:g:sampled}. We use  $\doublehat{\Psi}^m$ to represent samples of $\doublehat{\mu}^m_{i,j}(w)$. The symbol $\odot$ represents element-wise multiplication. 

Zeros in the denominator of \eqref{eq:doublehathscalar} are problematic for numeric analysis. We define the auxiliary function 
\begin{equation}
d(w_t,w_y) = \left(w_t\cos\xi_i+w_y\sin\xi_i\right)\left(w_t\cos\xi_j+w_y\sin\xi_j\right)\label{eq:Hdenominator}.
\end{equation}
However, filtering ensures $\doublehat{G}^m$ and $\doublehat{\Psi}^m$ are also zero when ${d(w_t,w_y)=0}$. Zeros in the numerator of \eqref{eq:doublehathscalar} also affect conditioning of the problem. Tikhonov regularization provides a generic mitigation strategy. Putting this together, we approximate the element-wise inverse of $\doublehat{H}$
\begin{equation}
[K]_{n,m} \coloneqq \cases{\frac{[\doublehat{H}^*]_{n,m}}{|[\doublehat{H}]_{n,m}|^2 + \epsilon}& $d(\frac{m}{N_t\Delta_t},\frac{n}{N_y\Delta_y})\neq0$\\
	0 & otherwise,}\label{eq:Htik}
\end{equation}
where $^*$ indicates complex conjugation, and $\epsilon$ is the smoothing parameter. This yields the estimate
\begin{equation}
\doublehat{\Psi}^m\approx\doublehat{G}^m\odot K.\label{eq:modreconstruct}
\end{equation}
Applying the 2D inverse discrete Fourier transform to the result we obtain a reconstruction of the filtered attenuation image. This process is described in Algorithm \ref{alg:brt_inv_mod_data}. The smoothing parameter $\epsilon$, in \eqref{eq:Htik}, can be adjusted for measurement noise and numerical errors. 
 
\begin{algorithm}
\caption{\textsc{BrtInvertFiltered}: Invert BRT data with bounded support. In this algorithm $\textsc{ComputeK}$ refers to the computation of $K$ using equations \eqref{eq:doublehathscalar}, \eqref{eq:Hsampled}, \eqref{eq:Hdenominator}, and \eqref{eq:Htik}. Here we use $\DFT^2$ and $\DFT^{-2}$ to represent the 2D discrete Fourier transform and its inverse, respectively.}
\label{alg:brt_inv_mod_data}
\begin{algorithmic}[1]
	\Require $G\in \mathbb{R}^{N_y\times N_t}$; $\Delta_t,\Delta_y,\epsilon\in \mathbb{R}^+$; $\xi_i,\xi_j\in(-\pi/2,\pi/2)$
	\Ensure $\Psi\in\mathbb{R}^{N_y\times N_t}$
	\State $K \leftarrow \textsc{ComputeK}(N_t,N_y,\Delta_t,\Delta_y,\xi_i,\xi_j,\epsilon)$
	\State $\doublehat{G} = \DFT^{2}\left\{G\right\}$
	\State $\doublehat{\Psi} =  \doublehat{G}\odot K$
	\State $\Psi = \DFT^{-2}\left\{\doublehat{\Psi}\right\}$
\end{algorithmic}
\end{algorithm}

Tikhonov regularization is generic and does not impose boundary conditions. For arbitrary angles $\xi_i$, and $\xi_j$, few samples of $\doublehat{\Psi}^m$ lie in the nullspace of the forward operator and it is sufficient to zero the results at these samples. Otherwise, it may be necessary to impose boundary conditions. For example, to ensure $\Psi^m$ is zero along the $t=0$ and $y=0$ boundary, all columns and all rows of $\doublehat{\Psi}^m$ must sum to 0.

%% file: results.tex
We provide results of numerical simulations to demonstrate the utility of this analysis. We use the modified Shepp-Logan phantom \cite{Toft1996,Shepp1974} in most of our simulations as depicted in Figure \ref{fig:brtdata}. This phantom is reasonably challenging and the BRT data can be determined analytically. For Figure \ref{fig:brtdata} we sample the image and data space uniformly in $y$ and $t$. For $y$ we use $N_y=600$ sampling over $[-1,1]$. For $t$ we use $N_t=400$ sampling over $[-0.75,0.75]$. This effects different sampling rates in $t$ and $y$. Limiting the extent of available BRT data in this way truncates the data both in $y$ and $t$ as shown in Figure \ref{fig:sheplogan_brt_data1} and Figure \ref{fig:sheplogan_brt_data2}

\begin{figure}
	\centering
	\setlength{\tabcolsep}{2pt}
	\begin{tabular}{ >{\centering\arraybackslash}m{1.4in} >{\centering\arraybackslash}m{1.4in} >{\centering\arraybackslash}m{1.4in} m{0.35in}}
		\subfloat[Reference image]{\includegraphics{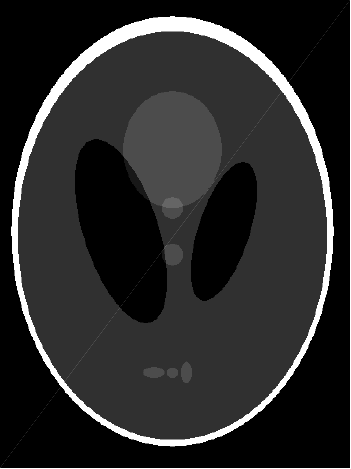}\label{fig:sheplogan_ref}}&
		\subfloat[BRT data, $\xi = \pi/11$]{\includegraphics{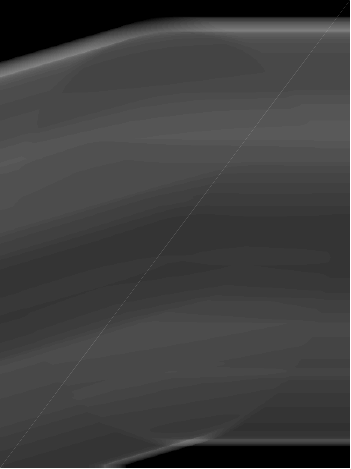}\label{fig:sheplogan_brt_data1}}&
		\subfloat[BRT data, $\xi = -\pi/5$]{\includegraphics{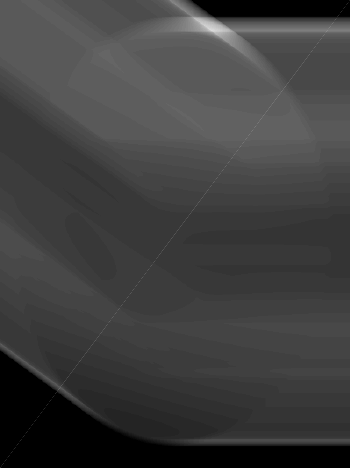}\label{fig:sheplogan_brt_data2}}&
		\raisebox{-1em}{\includegraphics{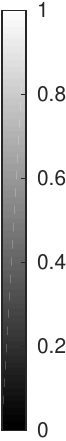}}
	\end{tabular}
	\caption{Reference image and associated BRT data. Figure \ref{fig:sheplogan_ref} depicts Shepp-Logan phantom as a reference image. Figure \ref{fig:sheplogan_brt_data1} and \ref{fig:sheplogan_brt_data1} depict BRT data with different scatter angles. The BRT data were determined analytically and sampled at the scatter points associated with the pixel centers of Figure \ref{fig:sheplogan_ref}.}\label{fig:brtdata}
\end{figure}

We first demonstrate filtering bounds support of the data. Results are shown for both the BRT and SBRT in Figure \ref{fig:brt_mod}. In this case the filtered image and filtered data were all obtained analytically and then sampled. 

\begin{figure}%
	\centering
	\setlength{\tabcolsep}{2pt}
	\begin{tabular}{@{} >{\centering\arraybackslash}m{2.11in} >{\centering\arraybackslash}m{2.11in}  m{0.35in}}
		\subfloat[Image I]{\includegraphics{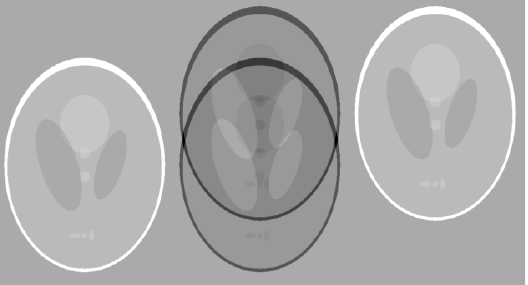}\label{fig:sheplogan_mod_ref}}&
		\subfloat[BRT data for Image I]{\includegraphics{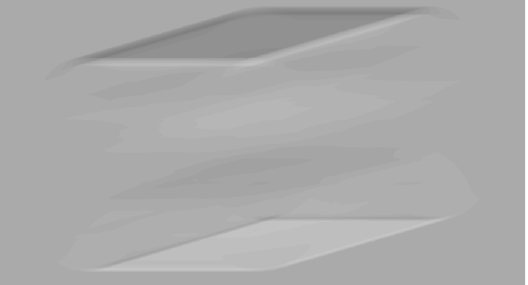}\label{fig:sheplogan_mod_brt}}&
		\multirow{2}{*}[1em]{\includegraphics{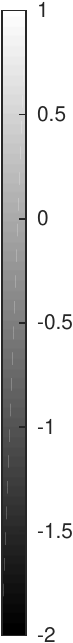}}\\
		
		\subfloat[Image II]{\includegraphics{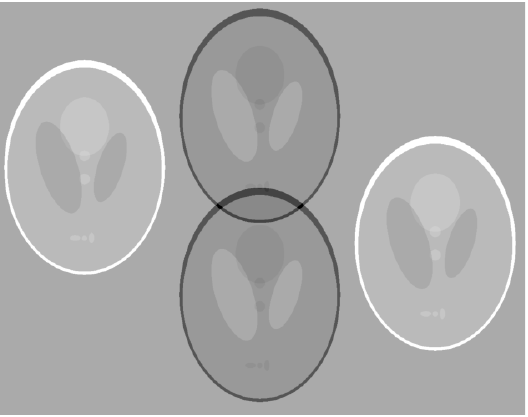}\label{fig:sheplogan_mod2_ref}}&
		\subfloat[SBRT data for Image II]{\includegraphics{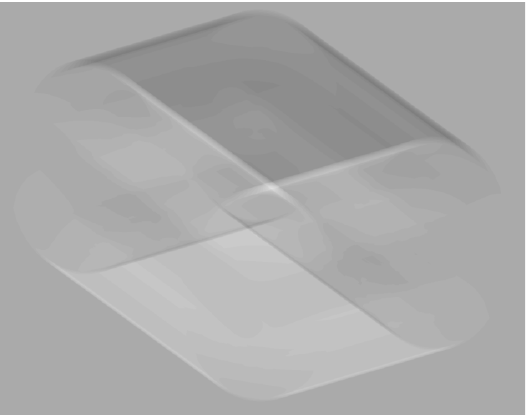}\label{fig:sheplogan_mod_sbrt}}
	\end{tabular}
	
	\caption{Image filtering effects BRT data with bounded support. Figure \ref{fig:sheplogan_mod_ref} depicts a notional phantom defined by filtering the image of Figure \ref{fig:sheplogan_ref} analytically using \eqref{eq:brt_mod_symmetric}. In this case $\xi_i=\pi$ and $\xi_j=\pi/11$ where the subscripts distinguish the directions ${\theta_i =(\cos\xi_i,\sin\xi_i)}$. The associated analytic BRT data are shown in Figure \ref{fig:sheplogan_mod_brt} and indicate bounded support. To bound support of SBRT data, filtering need only address the unique scatter directions associated with the two BRT data sets. Figure \ref{fig:sheplogan_mod2_ref} and Figure \ref{fig:sheplogan_mod_sbrt} show the filtered image and filtered SBRT data, respectively. Here the scatter angles for the BRT data composing the SBRT data are ${\xi_j \in \{\pi/11,-\pi/5\}}$.
	}\label{fig:brt_mod}
\end{figure}

Filtering can also be applied to sampled BRT data directly. For sampled data this effects small errors which we quantify against the reference data of Figure \ref{fig:brt_mod}. Results are shown in Figure \ref{fig:brt_mod_err}. Artifacts are observed at scatter points for which resulting rays are tangent to large transitions in the image. This is a consequence of sampling. For both BRT and SBRT filtering the peak absolute error is less than 5\% the peak image value. 

\begin{figure}%
	\centering
	\subfloat[Absolute BRT filtering error]{\includegraphics{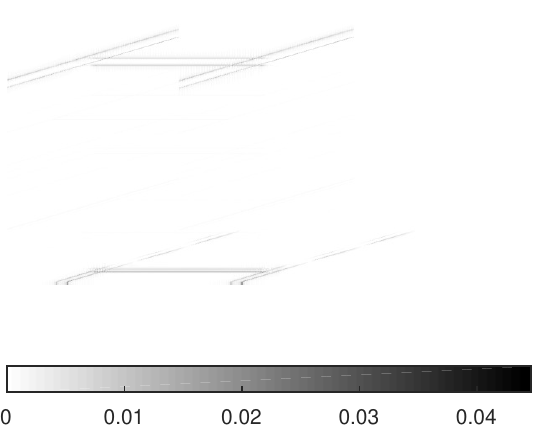}\label{fig:sheplogan_mod_err}}\quad
	\subfloat[Absolute SBRT filtering error]{\includegraphics{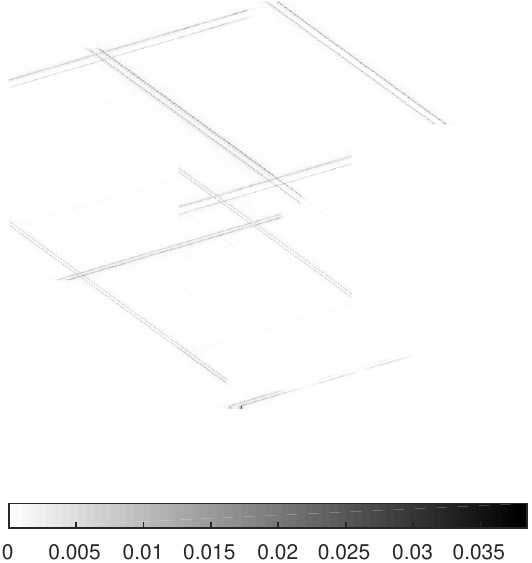}\label{fig:sheplogan_mod_sbrt_err}}
	
	\caption{Error extending and filtering sampled, truncated, BRT data. Figure \ref{fig:sheplogan_mod_err} depicts the error extending and filtering the data of Figure \ref{fig:sheplogan_brt_data1}. The reference data are shown in Figure \ref{fig:sheplogan_mod_brt}. Similarly, Figure \ref{fig:sheplogan_mod_sbrt_err} depicts the error extending and filtering the data of Figure \ref{fig:sheplogan_brt_data1} and Figure \ref{fig:sheplogan_brt_data2}. In this case the reference data are shown in Figure \ref{fig:sheplogan_mod_sbrt}.
	}\label{fig:brt_mod_err}
\end{figure}

Further analysis of $h(w)$ provides insights on BRT inversion. We can express \eqref{eq:h} in polar coordinates with the change of variables
\begin{equation}
w = \rho \left(\cos\phi,\sin\phi\right) \quad  \theta_i = \left(\cos\xi_i,\sin\xi_i\right) \quad  \theta_j = \left(\cos\xi_j,\sin\xi_j\right)
\end{equation}
such that
\begin{equation}
\doublehat{h}((\rho,\phi)) = \frac{-\cos\left(\phi - \frac{1}{2}\left(\xi_i+\xi_j\right)\right)\cos\left( \frac{1}{2}\left(\xi_i-\xi_j\right)\right)}{i\pi\rho \cos\left(\phi - \xi_i\right)\cos\left(\phi - \xi_j\right)}.\label{eq:doublehathpolar}
\end{equation}
We make a few observations. First, $\rho$ in the denominator of \eqref{eq:doublehathpolar} implies the BRT attenuates high frequency content. Reconstruction will be sensitive to noise at high frequencies. Second, there are singularities at ${\phi = \xi_i \pm \pi/2}$ and ${\phi = \xi_j \pm \pi/2}$. Filtering ensures the image and data are zero at these frequencies. Finally, \eqref{eq:doublehathpolar} is zero at ${\phi = \frac{1}{2}\left(\xi_i+\xi_j\right) \pm \pi/2}$. These zeros do not appear in the CBT, but arise in the combination of two CBTs. 

The matrix $K$ plays a critical role in BRT reconstruction \eqref{eq:modreconstruct}. This incorporates changes to $h(w)$ due to $\xi_j$, and the regularization term $\epsilon$. Changes to $|K|$ with respect to these terms is shown in Figure \ref{fig:Panalysis}. Here we fix ${\xi_i = \pi}$  without loss of generality. The lines indicating strong attenuation are due to singularities of \eqref{eq:doublehathpolar} at ${\phi = \pm \pi/2}$, and ${\phi = \xi_j \pm \pi/2}$. Zeros in \eqref{eq:doublehathpolar} effect large amplitudes in $|K|$ along ${\phi = \xi_j/2}$. However, this amplitude is curtailed through regularization as $\epsilon$ increases. Without regularization, we would expect reconstruction artifacts along this spectral line.

\begin{figure}%
	\centering
	\setlength{\tabcolsep}{2pt}
	\begin{tabular}{>{\centering\arraybackslash}m{1em} >{\centering\arraybackslash}m{1.36in} >{\centering\arraybackslash}m{1.36in} >{\centering\arraybackslash}m{1.36in} m{0.35in}}
		& $\xi_j= \pi/20$ & $\xi_j = \pi/7$ & $\xi_j = \pi/4$ &\\
		
		\rotatebox{90}{$\epsilon=1\mathrm{e}^{-6}$}&
		\includegraphics{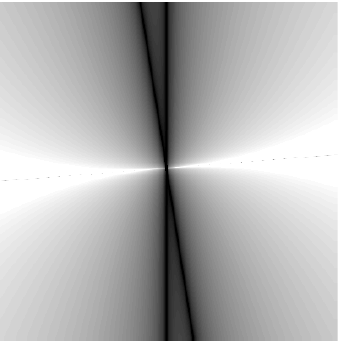}&
		\includegraphics{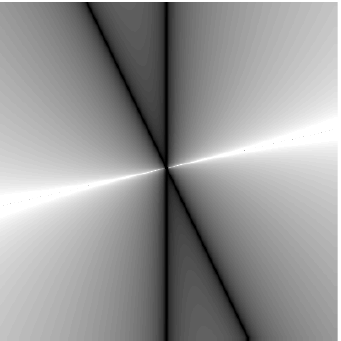}&
		\includegraphics{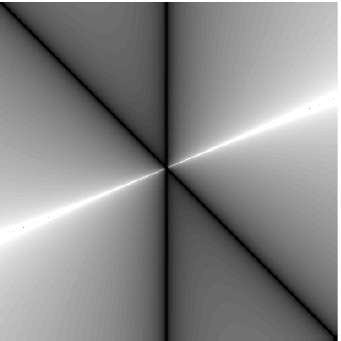}&
		\multirow{3}{*}[1.5em]{\includegraphics{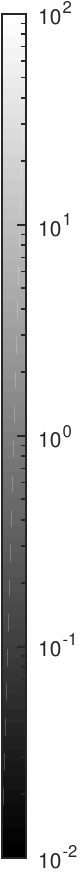}}\\	
		
		\rotatebox{90}{$\epsilon=1\mathrm{e}^{-5}$} &
		\includegraphics{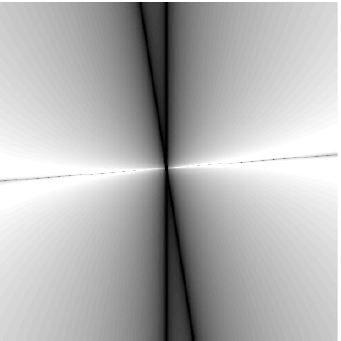}&
		\includegraphics{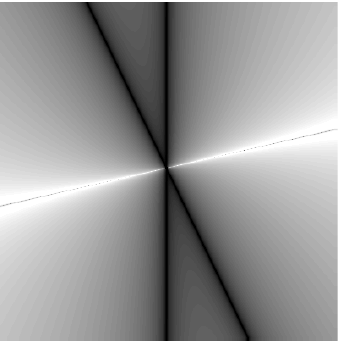}&
		\includegraphics{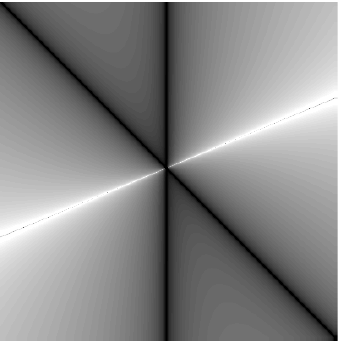}\\	
		
		\rotatebox{90}{$\epsilon=1\mathrm{e}^{-4}$} &
		\includegraphics{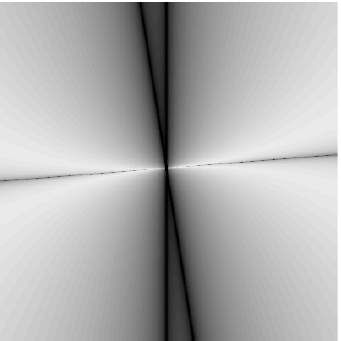}&
		\includegraphics{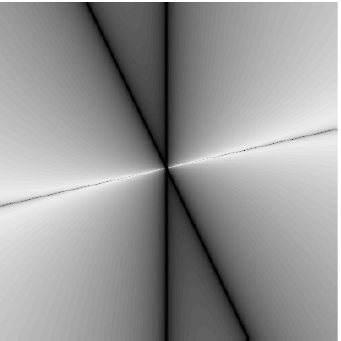}&
		\includegraphics{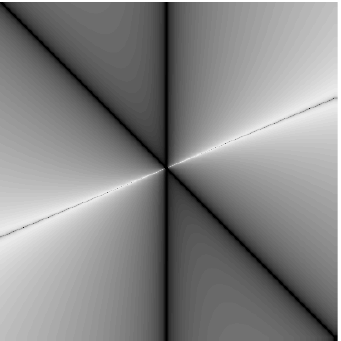}
	\end{tabular}
	
	\caption{Changes in $|K|$  \eqref{eq:Htik} with respect to regularization $\epsilon$ and angle $\xi_j$. The first column of images corresponds to scatter direction $\xi_j=\pi/20$, the second column $\xi = \pi/7$, and the third column $\xi = \pi/4$. For all images we fix ${\xi_i = \pi}$. Each row of images uses a different $\epsilon$; the first row uses $1\mathrm{e}^{-6}$, the second $1\mathrm{e}^{-5}$, and the third $1\mathrm{e}^{-4}$. For all images the zero-frequency content is centered for both axes. Further, the same display scale is used as shown in the colorbar.
	}\label{fig:Panalysis}
\end{figure}

The original global BRT inversion formula is due to Florescu et al. \cite{Florescu2011}. We will refer to this as the FMS formula. Specifically contrasting with our algorithm, we analyze the same square phantom in Figure \ref{fig:FlorescuContrast}. The original work assumed data available over an infinite strip with no additional insights on limiting the data. The data of Figure \ref{fig:flor_brt_data} violates this assumption. Directly applying the FMS formula to this data yields poor results as shown in Figure \ref{fig:flor_inv_truncated}. However, we can simulate additional data using Algorithm \ref{alg:cbt_extend}. Applying the FMS formula to the extended BRT data yields results consistent with those previously published \cite{Florescu2011}. In this way, Algorithm \ref{alg:cbt_extend} can be used as a preprocessing step to reduce the extent of sampling required for reconstruction using the FMS formula. The direction of the artifacts is explained by the nullspace of the forward operator \eqref{eq:h}. Striations are observed in the direction ${\xi/2+\pi/2}$. Regularization further improves reconstruction as demonstrated in Figure \ref{fig:flor_inv_regularized}.

\begin{figure}%
	\centering
	\setlength{\tabcolsep}{2pt}
	\begin{tabular}{ >{\centering\arraybackslash}m{1.4in} >{\centering\arraybackslash}m{1.4in} >{\centering\arraybackslash}m{1.4in} m{0.35in}}
\subfloat[Reference image]{\includegraphics[]{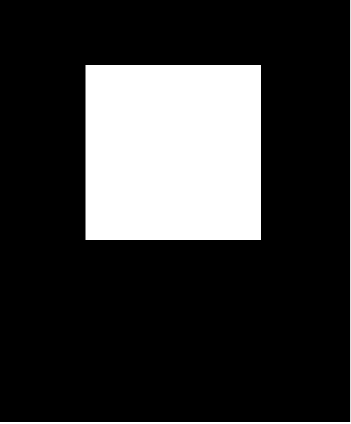}\label{fig:flor_ref}}&
\subfloat[BRT data]{\includegraphics[]{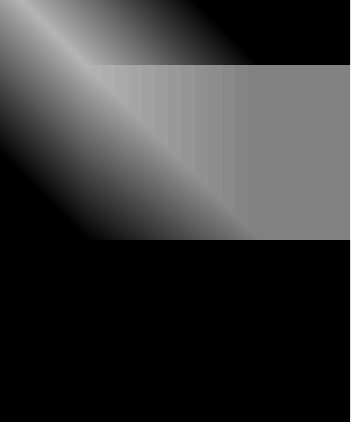}\label{fig:flor_brt_data}}&
\subfloat[Results, FMS formula]{\includegraphics[]{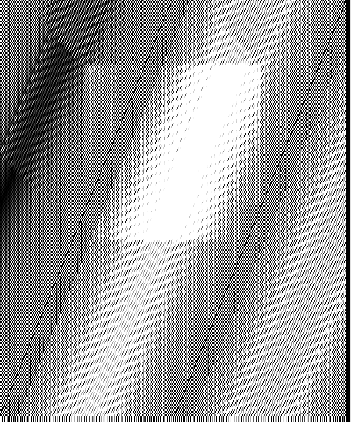}\label{fig:flor_inv_truncated}}&
		\multirow{2}{*}[1.5em]{\includegraphics{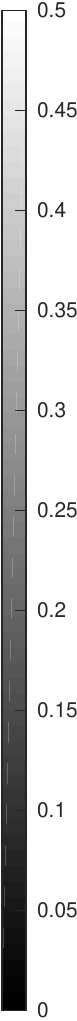}}\\	
		
\subfloat[Extended BRT data]{\includegraphics[]{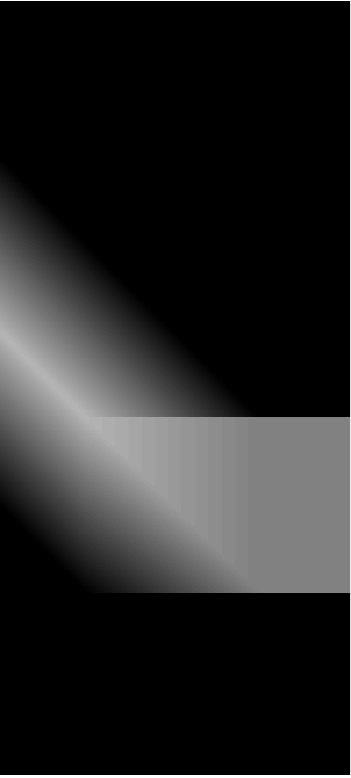}\label{fig:flor_brt_extend}}&
\subfloat[Results, FMS formula]{\includegraphics[]{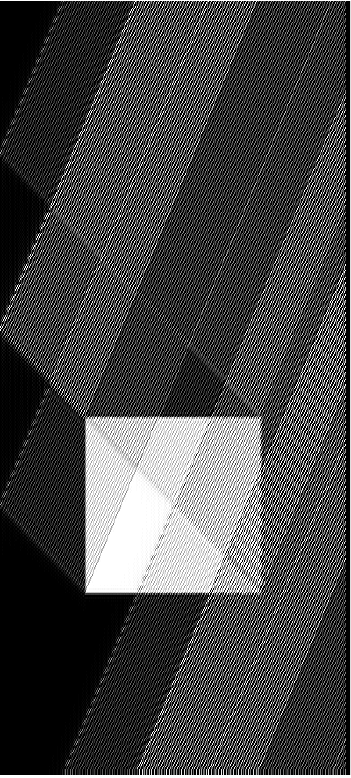}\label{fig:flor_inv_extended}}&
\subfloat[Results, Algorithm \ref{alg:brt_inv_mod_data}]{\includegraphics[]{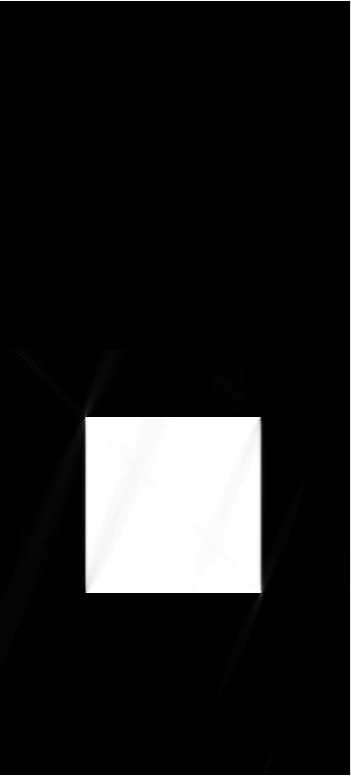}\label{fig:flor_inv_regularized}}
	\end{tabular}
		\caption{Noise-free reconstruction from limited data. The reference image is shown in Figure \ref{fig:flor_ref}, and we limit the available BRT data as shown in Figure \ref{fig:flor_brt_data} with $\xi_j=-\pi/4$. FMS \cite{Florescu2011} reconstruction, using limited data, is shown in Figure \ref{fig:flor_inv_truncated}. The limited BRT data of Figure \ref{fig:flor_brt_data} can be extended using Algorithm \ref{alg:cbt_extend} as shown in Figure \ref{fig:flor_brt_extend}. Figure \ref{fig:flor_inv_extended} depicts results applying the FMS formula to the extended data of Figure \ref{fig:flor_brt_extend}. Similarly, Figure \ref{fig:flor_inv_regularized} depicts results applying Algorithm \ref{alg:brt_inv_mod_data} to the extended data of Figure \ref{fig:flor_brt_extend}. All images use the same display scale shown in the colorbar.
	}\label{fig:FlorescuContrast}
\end{figure}

Inversion results for the Shepp-Logan phantom on noisy data are shown in Figure \ref{fig:BrtInvChanges}. The BRT data were obtained analytically, sampled, and corrupted with additive Gaussian noise. For small $\epsilon$, we see artifacts where the direction $\xi$ is tangent to high frequency edges of the image. This is a consequence of sampling errors and extending the BRT data. Additionally, edges perpendicular to the direction $\xi/2$ are not well resolved. This effects blurring along the direction $\xi/2$. Increasing $\epsilon$ increases the angular extent of blurring. The effect is reduced as $\xi$ increases.

\begin{figure}%
	\centering
	\setlength{\tabcolsep}{2pt}
	\begin{tabular}{>{\centering\arraybackslash}m{1em} >{\centering\arraybackslash}m{1.36in} >{\centering\arraybackslash}m{1.36in} >{\centering\arraybackslash}m{1.36in} m{0.35in}}
		& $\xi_j = \pi/20$ & $\xi_j = \pi/7$ & $\xi_j = \pi/4$ &\\
		
		\rotatebox{90}{$\epsilon=1\mathrm{e}^{-6}$}&
		\includegraphics{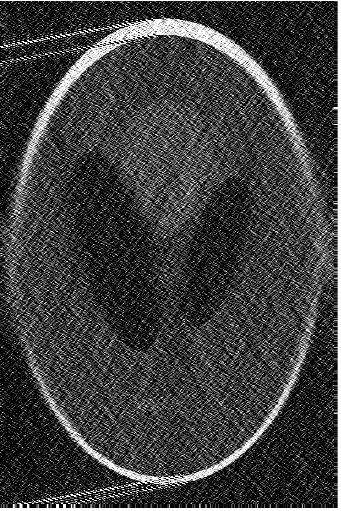}&
		\includegraphics{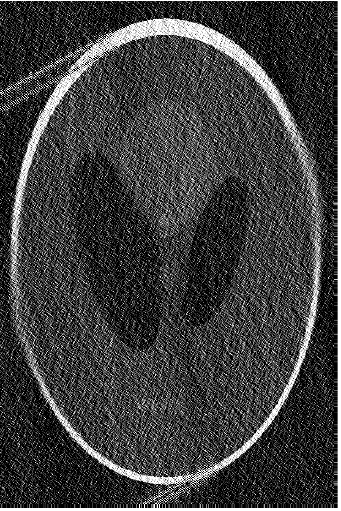}&
		\includegraphics{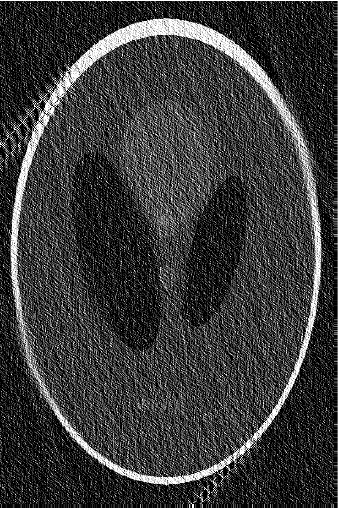}&
		\multirow{3}{*}[1.5em]{\includegraphics{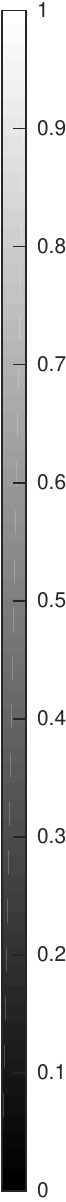}}\\	
		
		\rotatebox{90}{$\epsilon=1\mathrm{e}^{-5}$} &
		\includegraphics{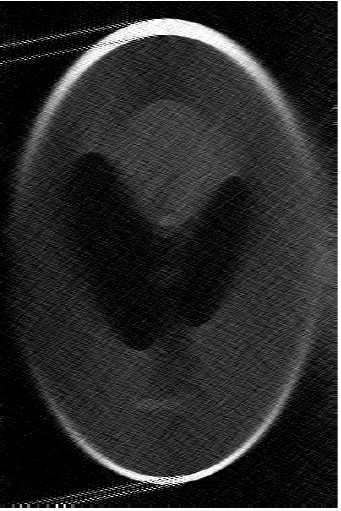}&
		\includegraphics{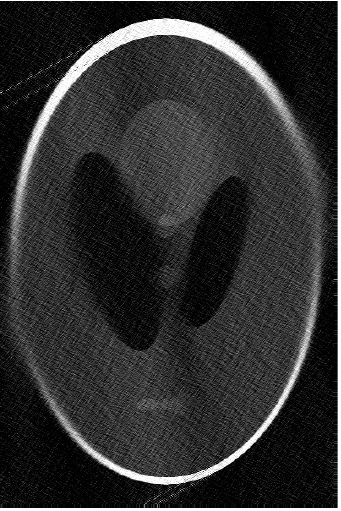}&
		\includegraphics{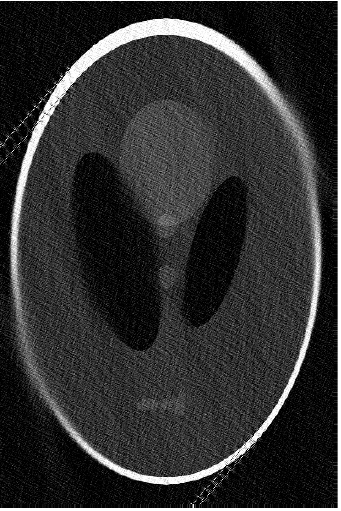}\\	
		
		\rotatebox{90}{$\epsilon=1\mathrm{e}^{-4}$} &
		\includegraphics{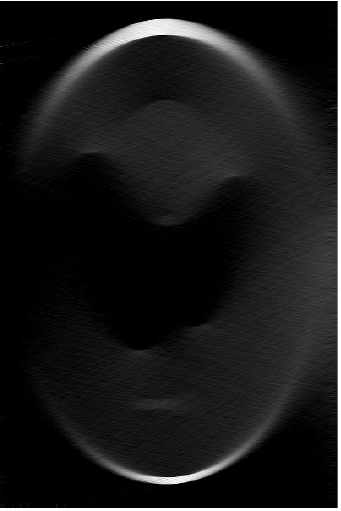}&
		\includegraphics{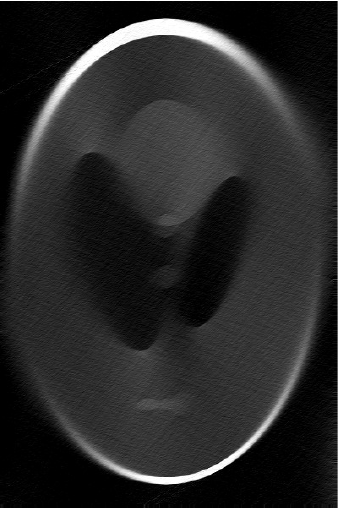}&
		\includegraphics{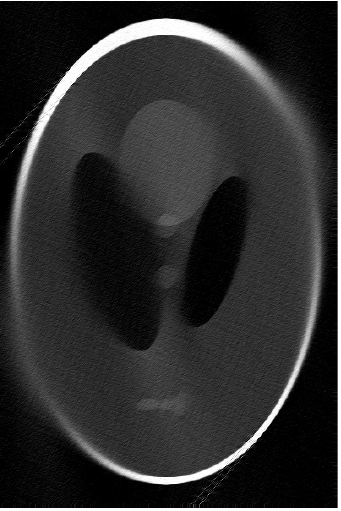}
	\end{tabular}
	\caption{Reconstruction of noisy, truncated, BRT data using Algorithm \ref{alg:brt_inv_mod_data}. The first column of images corresponds to scatter direction $\xi_j=\pi/20$, the second column $\xi_j = \pi/7$, and the third column $\xi_j = \pi/4$. Each row of images uses a different $\epsilon$ in \eqref{eq:Htik} which appears in the reconstruction formula \eqref{eq:modreconstruct}; the first row uses $1\mathrm{e}^{-6}$, the second $1\mathrm{e}^{-5}$, and the third $1\mathrm{e}^{-4}$. All images use the same display scale shown in the colorbar. The same realization of Gaussian noise was added to each data set. The standard deviation was $10^{-3}$ times the peak amplitude of the image.
	}\label{fig:BrtInvChanges}
\end{figure}

%% file: cbt_Fourier.tex
We define the Fourier transform as a function of frequency to avoid scaling the inverse. For a one-dimensional function we define the one-dimensional Fourier transform and its inverse
\begin{eqnarray}
\mathcal{F}^1\{f(x)\} &:=\int_{-\infty}^{+\infty} f(x)e^{-i2\pi w x}dx\\
\mathcal{F}^{-1}\{\hat{f}(w)\} &:=\int_{-\infty}^{+\infty} \hat{f}(w)e^{i2\pi w x}dw
\end{eqnarray}
For two-dimensional functions we define the two-dimensional Fourier transform and its inverse
\begin{eqnarray}
\mathcal{F}^2\{f(x)\} &:= \int_{\mathbb{R}^2} f(x)e^{-i2\pi w\cdot x}d^2x\\
\mathcal{F}^{-2}\{\doublehat{f}(w)\} &:=\int_{\mathbb{R}^2} \doublehat{f}(w)e^{i2\pi w\cdot x}d^2w.
\end{eqnarray}
In this form, we have
\begin{eqnarray}
\mathcal{F}^1\{u(x)\}&=\frac{1}{i2\pi w}+\frac{1}{2}\delta(w)\label{eq:ft:u}
\end{eqnarray}
where $\delta(x)$ and $u(x)$  represent the Dirac delta function and the unit step function, respectively.  

For CBT data associated with a fixed direction, $\theta$, we define the two-dimensional Fourier transform
\begin{eqnarray}
\doublehat{b}_\theta(w) &= \mathcal{F}^2\{(B \mu_C)(x,\theta)\}\\
&= \int_{\mathbb{R}^2} \int_{0}^\infty \mu_C (x+t\theta) dt\, e^{-i2\pi w\cdot x}d^2x\\
&= \int_{0}^\infty \int_{\mathbb{R}^2} \mu_C (y) e^{-i2\pi w\cdot y} d^2y\, e^{i2\pi t w\cdot \theta}dt\label{eq:cbt_ft:2}\\
&= \doublehat{\mu}_C(w)\int_0^\infty e^{i2\pi t w\cdot \theta}dt\label{eq:cbt_ft:3}\\
&= \doublehat{\mu}_C(w)\left[\frac{-1}{i2\pi w\cdot\theta}+\frac{1}{2}\delta(w\cdot\theta)\right].\label{eq:cbt_ft:rslt}
\end{eqnarray}
In \eqref{eq:cbt_ft:2} we changed the order of integration and substituted ${y=x+t\theta}$. In \eqref{eq:cbt_ft:3} we substituted ${\doublehat{\mu}_C(w)=\mathcal{F}^2\{\mu_C(x)\}}$. Finally, in \eqref{eq:cbt_ft:rslt} we made use of \eqref{eq:ft:u}.

%% file: brt_inverse_fourier.tex
To invert the BRT, we start by multiplying both sides of \eqref{eq:brt_ft} with the reciprocal of \eqref{eq:h}. Rearranging terms, we have
\begin{eqnarray}
\fl\doublehat{\mu}_C(w) =&i\pi\doublehat{\mu}_C(w)\frac{\left(w\cdot\theta_i\right)\left(w\cdot\theta_j\right)}{ w\cdot\left(\theta_i+\theta_j\right)}\delta\left(w\cdot\theta_i\right) +i\pi\doublehat{\mu}_C(w)\frac{\left(w\cdot\theta_i\right)\left(w\cdot\theta_j\right)}{ w\cdot\left(\theta_i+\theta_j\right)}\delta\left(w\cdot\theta_j\right)\nonumber\\
\fl&+\doublehat{g}_{i,j}(w)\frac{-i2\pi \left(w\cdot\theta_i\right)\left(w\cdot\theta_j\right)}{w\cdot\left(\theta_i+\theta_j\right)},\quad \forall\, w\notin \Theta_{i,j}.\label{eq:brt_inv_system}
\end{eqnarray}
The first two terms on the right hand side vanish under integration.  We note the inverse two-dimensional Fourier transform
\begin{equation}
\int_{\mathbb{R}^2}\left(w\cdot\theta\right)\delta\left(w\cdot\theta\right)e^{i2\pi w\cdot x}d^2w = 0,\quad \forall\, x\in\mathbb{R}^2
\end{equation} 
Incorporating multiplicative functions does not change this result as long as they are finite for all $w\cdot\theta=0$. For the first term in \eqref{eq:brt_inv_system}, we expand ${w = s\theta_i + t\theta_i^\perp}$ as an orthonormal basis and set $s=0$. This leads to
\begin{equation}
\doublehat{\mu}_C(w)\left.\frac{w\cdot\theta_j}{ w\cdot\left(\theta_i+\theta_j\right)}\right|_{w=t\theta_i^\perp}=\doublehat{\mu}_C(t\theta_i^\perp),\quad \forall\, \theta_i\neq\theta_j\label{eq:frac}
\end{equation}
which is finite for all $t$ by our assumptions on $\mu_C(x)$. Applying a similar process for the second term, we find
\begin{equation}
\mathcal{F}^{-2}\left\{i\pi\doublehat{\mu}_C(w)\frac{\left(w\cdot\theta_i\right)\left(w\cdot\theta_j\right)}{ w\cdot\left(\theta_i+\theta_j\right)}\left[\delta\left(w\cdot\theta_i\right)+\delta\left(w\cdot\theta_j\right)\right]\right\}=0.
\end{equation}
To address the third term in \eqref{eq:brt_inv_system} we make use of the derivative and integral properties of the Fourier transform. Now we expand $x = s\theta + t\theta^\perp$ and consider the directional derivative
\begin{equation}
\left.\frac{d}{d\theta} f(x)\right|_{x=s\theta + t\theta^\perp} = \frac{d}{ds}f(s\theta + t\theta^\perp).
\end{equation}
Applying this to the inverse two-dimensional Fourier transform, we find
\begin{eqnarray}
\frac{d}{d\theta} f(x) &=\frac{d}{d\theta} \mathcal{F}^{-2}\left\{ \doublehat{f}(w) \right\}\\
&= \int_{\mathbb{R}^2}\frac{d}{ds}\doublehat{f}(w)e^{i2\pi \left(s w\cdot\theta + t w\cdot\theta^\perp\right)}d^2w\\
&= \mathcal{F}^{-2}\left\{ i2\pi \left(w\cdot\theta\right) \doublehat{f}(w)\right\}.
\end{eqnarray}
We previously derived the integration property of the two-dimensional Fourier transform in the context of the CBT. From \eqref{eq:cbt_ft:rslt}, we have
\begin{equation}
\mathcal{F}^{2}\left\{\int_{0}^{\infty}f\left(x + s\theta\right)ds\right\} =\doublehat{f}(w)\left[\frac{-1}{i2\pi w\cdot\theta}+\frac{1}{2}\delta(w\cdot\theta)\right].\label{eq:fourier_integral}
\end{equation}
When $\doublehat{f}(w\cdot\theta)=0$ for all $w\in\mathbb{R}^2$, substituting $-\theta$ in \eqref{eq:fourier_integral} we also find
\begin{eqnarray}
\int_{0}^{\infty}f\left(x + s\theta\right)ds &=-\int_{0}^{\infty}f\left(x + s(-\theta)\right)ds \\
&=-\int_{-\infty}^{0}f\left(x + s\theta\right)ds.
\end{eqnarray}
From \eqref{eq:brt_ft}, $\doublehat{g}_{i,j}(0)$ is not guaranteed to be finite, much less zero. However, $\left(w\cdot\theta_i\right)\left(w\cdot\theta_j\right)\doublehat{g}_{i,j}(w)=0$ for all $w\cdot\left(\theta_i+\theta_j\right)=0$. Putting this all together, we have equivalent expressions
\begin{eqnarray}
\mu_C(x)&= \frac{1}{\|\theta_i+\theta_j\|}\int_{0}^\infty \frac{d}{d\theta_i}\frac{d}{d\theta_j}g_{i,j}\left(x + s\frac{\theta_i+\theta_j}{\|\theta_i+\theta_j\|}\right)ds\\
&= \frac{-1}{\|\theta_i+\theta_j\|}\int_{-\infty}^0 \frac{d}{d\theta_i}\frac{d}{d\theta_j} g_{i,j}\left(x + s\frac{\theta_i+\theta_j}{\|\theta_i+\theta_j\|}\right)ds.
\end{eqnarray}
We emphasize equality only holds for images, $\mu_C(x)$, with bounded support. The assumption is necessary due to the nullspace of the forward operator.

%% file: parallelogram.tex
Parallelograms are often expressed in terms of the edge directions and edge lengths. We consider the directions $\theta_i$, $\theta_j$ and associated edge lengths $a_i$, $a_j$, respectively. The total area of the resulting parallelogram is ${a_i a_j |\det\left(\theta_i,\theta_j\right)|}$. As an alternative to edge lengths, we also consider the orthogonal distance between parallel sides. We define $b_i$ and $b_j$ as the extent (height) of the parallelogram in the orthogonal directions $\theta_i^\perp$ and $\theta_j^\perp$, respectively. These distances are related to the edge lengths through the change of variables
\begin{eqnarray}
b_i &\coloneqq a_j  \left|\det\left(\theta_i,\theta_j\right)\right| \label{eq:a2vi}\\
b_j &\coloneqq a_i \left|\det\left(\theta_i,\theta_j\right)\right|\label{eq:a2vj}.
\end{eqnarray}
Additionally, we define the one-dimensional rectangular function
\begin{equation}
\Pi_T(t) \coloneqq \cases{ 1, & $|t|\leq T/2$\\0,& otherwise,}\label{eq:rectfunc}
\end{equation}
We define the two-dimensional parallelogram indicator function, centered at $x=0$,
\begin{equation}
p_{i,j}(x;a_i,a_j) \coloneqq \Pi_{b_i}\left(x\cdot\theta_i^\perp\right)\Pi_{b_j}\left(x\cdot\theta_j^\perp\right).\label{eq:p}
\end{equation}
Here $b_i$, $b_j$ are determined by $a_j$, $a_i$ using \eqref{eq:a2vi} and \eqref{eq:a2vj}, respectively. The area of this function is equivalently
\begin{equation}
\int_{R^2} p_{i,j}(x;a_i,a_j) d^2x =  a_i a_j |\det\left(\theta_i,\theta_j\right)| =  \frac{b_i b_j}{|\det\left(\theta_i,\theta_j\right)|}.\label{eq:p:area}
\end{equation}

To determine the two-dimensional Fourier transform of \eqref{eq:p}, we exploit the convolution property of the Fourier transform. We transform the two rectangular functions separately, then convolve the results in the frequency domain. The one-dimensional Fourier transform of \eqref{eq:rectfunc}
\begin{equation}
\mathcal{F}\left\{\Pi_T(t)\right\}=T\sinc\left(w T\right).
\end{equation}
Extending this to two dimension, we have the relation 
\begin{eqnarray}
\mathcal{F}^{2}\left\{\Pi_{T}(x\cdot\theta)\delta\left(x\cdot\theta^\perp\right) \right\}  
&=T\sinc(T w\cdot\theta)\label{eq:rectslice}\\
\mathcal{F}^{2}\left\{\Pi_{T}(x\cdot\theta)\right\}  
&=T\sinc(T w\cdot\theta)\delta\left(w\cdot\theta^\perp\right).
\end{eqnarray}
We derive the two-dimensional Fourier transform of \eqref{eq:p}
\begin{eqnarray}
\fl \doublehat{p}_{i,j}\left(w;a_i,a_j\right) &= b_i b_j \sinc(b_i w\cdot\theta_i^\perp)\delta\left(w\cdot\theta_i\right)*\sinc(b_j w\cdot\theta_j^\perp)\delta\left(w\cdot\theta_j\right)\\
\fl &= b_i b_j\int_{\mathbb{R}^2} \sinc(b_i y\cdot\theta_i^\perp)\delta\left(y\cdot\theta_i\right)\sinc(b_j \left(w-y\right)\cdot\theta_j^\perp)\nonumber\\
\fl &\phantom{= b_i b_j}\quad\times\delta\left(\left(w-y\right)\cdot\theta_j\right)d^2y\label{eq:pft:wide}\\
\fl &=  b_i b_j \int_{\mathbb{R}} \sinc(b_i t) \sinc(b_j \left(w-t\theta_i^\perp\right)\cdot\theta_j^\perp)\delta\left(\left(w-t\theta_i^\perp\right)\cdot\theta_j\right)dt\label{eq:pft:intt}\\
\fl &= \frac{b_i b_j}{|\det\left(\theta_i,\theta_j\right)|}\sinc\left(b_i\frac{w\cdot\theta_j}{\theta_i^\perp\cdot\theta_j}\right)\sinc\left(b_j w\cdot\left(\theta_j^\perp - \theta_j\frac{\theta_i\cdot\theta_j}{\theta_i^\perp\cdot\theta_j}\right)\right)\label{eq:pft:sincext}\\
\fl &=  \frac{b_i b_j}{|\det\left(\theta_i,\theta_j\right)|}\sinc\left(b_i\frac{w\cdot\theta_j}{\theta_i^\perp\cdot\theta_j}\right)\sinc\left(b_j\frac{w\cdot\theta_i}{\theta_i\cdot\theta_j^\perp}\right)\label{eq:pft:v}\\
\fl &= a_i a_j |\det\left(\theta_i,\theta_j\right)|\sinc\left(a_j w\cdot\theta_j\right)\sinc\left(a_i w\cdot\theta_i\right)\label{eq:pft:a}.
\end{eqnarray}
For \eqref{eq:pft:intt}, we expand the integration variable in \eqref{eq:pft:wide} using the orthonormal basis, ${y = s\theta_i + t\theta_i^\perp}$, and integrate over $s$. Changing the variable of integration again effects a change in scaling in \eqref{eq:pft:sincext}. The second $\sinc$ function of \eqref{eq:pft:sincext} comprises expansion of $\theta_i$ using the orthonormal basis $\theta_j$, $\theta_j^\perp$. Restoring $\theta_i$, we obtain \eqref{eq:pft:v}. Restoring $a_j$, $a_i$ using \eqref{eq:a2vi}, \eqref{eq:a2vj}, we obtain \eqref{eq:pft:a}.

%% file: noninteger_shift.tex
Non-integer shifts of sampled signals requires interpolation. Fast implementations of the discrete Fourier transform (DFT) can be leveraged to perform this task quickly. For continuous signal $x(t)$, and uniform sample spacing $\Delta$, we define the $N$-length sampled signal
\begin{equation}
x[n] \coloneqq x(\Delta n), \quad \forall\, n=\{0,\ldots,N-1\}.
\end{equation}
The Fourier coefficients are given using the DFT
\begin{equation}
y[m] = \sum_{n=0}^{N-1} x[n]\exp\left(-i2\pi n m /N\right)
\end{equation}
for $m\in\{0,\ldots, N-1\}$. The corresponding inverse DFT is
\begin{eqnarray}
x[n] &= \DFT^{-1}\left\{y[m]\right\}\\
&=\frac{1}{N}\sum_{m=0}^{N-1} y[m]\exp\left(i2\pi n m /N\right)\label{eq:idtft}
\end{eqnarray}
Using the generative property of the Fourier series representation\cite{Proakis2006}, we can approximate 
\begin{eqnarray}
x(\Delta n - \Delta s) &\approx  \frac{1}{N}\sum_{m=0}^{N-1} y[m]\exp\left(i2\pi (n-s) m /N\right)\\
&=\DFT^{-1}\left\{y[m]\exp\left(-i2\pi s m /N\right)\right\}.\label{eq:shifted}
\end{eqnarray}
Since $s$ is represented in samples, equation \eqref{eq:shifted} is independent of sampling rate. This is particularly efficient when multiple shifted copies of the same signal are required. In such cases $y[m]$ need only be computed once. Additional savings are realized computing the $\DFT^{-1}$ in \eqref{eq:shifted} for all signals at once. This process is described in Algorithm \ref{alg:nonintshift}. Here we have included additional inputs indicating zero padding, $p$, and fill samples ${\bf f}$ to reduce aliasing.

\begin{algorithm}
	\caption{\textsc{NonIntShift}: Non-integer shifting of a sampled signal. We use $\odot$ to represent element-wise multiplication with assumed expansion along singleton dimensions. When ${\bf x}$ is matrix-valued and ${\bf s}$ is scalar-valued, the shift will be applied to all columns of ${\bf x}$ independently. We use $0^{p}$ to represent the $p$-length column vector of all zeros.}
	\label{alg:nonintshift}
	\begin{algorithmic}[1]
		\Require ${\bf x}$, ${\bf s}$, $p$, ${\bf f}$
		\Ensure  $Z$
		\State $x = \vertcat({\bf x},\, 0^{p},\, {\bf f})$
		\State ${\bf n} = \left[\matrix{0 & 1 & \cdots & N-1}\right]^T$
		\State $y = \DFT(x)$
		\State $W = \exp\left(-i2\pi\left( {\bf n} \odot {\bf s}^T\right)/N\right)$
		\State $Z = \DFT^{-1} \{y \odot W\}$
	\end{algorithmic}
\end{algorithm}